\documentclass[11pt]{article} 

\usepackage[utf8]{inputenc} 

\usepackage{geometry} 
\usepackage{enumerate}
\usepackage{amsmath}
\usepackage{amssymb}
\usepackage{dsfont}
\usepackage{verbatim}
\usepackage{amsthm}
\usepackage{graphicx}

\usepackage{amsfonts}
\usepackage{xcolor}
\usepackage{caption}
\usepackage{subcaption}
\usepackage{placeins}
\usepackage{hyperref}

\theoremstyle{plain}
\newtheorem{lemma}{Lemma}[section]

\numberwithin{equation}{section}
\numberwithin{figure}{section}
\numberwithin{table}{section}

\title{Calibrating Local Volatility Models with Stochastic Drift and
Diffusion\footnote{Electronic version of the article published as International
Journal of Theoretical and Applied Finance, Volume No. 25, Issue No. 02, Article
No. 2250011, Year 2022 \url{https:///dx.doi.org/10.1142/S021902492250011X}
\copyright World Scientific Publishing Company
\url{https://www.worldscientific.com/worldscinet/ijtaf}}\vspace{1em}}
\author{\Large Orcan \"Ogetbil\footnote{orcan.ogetbil@wellsfargo.com}, \ 
Narayan Ganesan\footnote{narayan.ganesan@wellsfargo.com}, \ 
and Bernhard Hientzsch\footnote{bernhard.hientzsch@wellsfargo.com}\\
{\small Corporate Model Risk, Wells Fargo Bank}}

\date{} 

\begin{document}

\maketitle
\setcounter{secnumdepth}{3}
\begin{abstract}

We propose Monte Carlo calibration algorithms for three models: local volatility
with stochastic interest rates, stochastic local volatility with deterministic
interest rates, and finally stochastic local volatility with stochastic interest
rates. For each model, we include detailed derivations of the corresponding SDE
systems, and list the required input data and steps for calibration. We give
conditions under which a local volatility can exist given European option
prices, stochastic interest rate model parameters, and correlations. 
The models are posed in a foreign exchange setting. 
The drift term for the exchange rate is given as a difference of two stochastic
short rates, domestic and foreign, each modeled by a Gaussian one-factor model
with deterministic shift (G1++) process.
For stochastic volatility, we model the variance for the exchange rate by a
Cox-Ingersoll-Ross (CIR) process. We include tests to show
the convergence and the accuracy of the proposed algorithms.

\end{abstract}

\section{Introduction}

In quantitative finance, considerable amount of research focuses on modeling the
market-observed smile of the implied volatility surface.
There are competing approaches to tackle this problem. Most notably, stochastic
volatility and local volatility models are often applied in practice to imitate
the market-observed smile. Stochastic volatility models aim to capture the
volatility dynamics observed in the market. While such models often capture 
the implied volatilities in certain tenor and strike ranges well, there are
often ranges that are not repriced well. Moreover, the parametric structure of
these models makes proper calibration computationally challenging.
The local volatility models are simpler due to their non-parametric nature; and
they are constructed to fit completely arbitrage-free implied volatility
surfaces. However, they fail to capture the proper dynamics of implied
volatilities, as they construct a static local volatility surface, which
inherently assumes deterministic spot volatility for the underlying process.
These shortcomings lead practitioners to combine them in unified frameworks,
called \emph{stochastic local volatility} models, with the intention of
combining the advantages of both approaches. The combined model is of special
interest to the financial industry, as it would provide a methodology to work
with a more complete set of risk factors associated with exotic instruments,
such as options on exchange rates, foreign stocks, quantoes and baskets.

Dupire's local volatility model \cite{Dupire1994,
DermanKani1994} constructs a unique diffusion process that is consistent with
the European vanilla option market quotes. Being calibrated to vanilla option
quotes, the local volatility surface has become a standard tool to capture the
risk-neutral marginal distributions of the underliers implied by market European
option price quotes, and is being utilized by
practitioners for pricing and risk-assessment of more exotic instruments. In its
original formulation, the local volatility model assumes deterministic drift and
diffusion terms. In our work, we relax both of these constraints in turn, and
study the generalizations of the local volatility model first under a drift that
is the difference between two stochastic short rates, then under stochastic
diffusion, and eventually under both stochastic drift and stochastic diffusion.
While we base our set up in foreign exchange context with the foreign exchange
rate, the domestic and base short rates as the risk drivers, our results can be
applied in equity or commodity contexts. The theoretical setup and algorithms
presented in the paper can be simplified with minimal effort to the cases where
the drift is modeled more simplistically with a single short rate driver.

The theoretical framework for extending the local volatility model with
a single stochastic rate was presented in \cite{Atlan2006}. Calibration of this
model using finite difference and Monte Carlo (MC) methods was discussed in
\cite{Hu2015, HokTan2018}. As a further extension, \cite{Deelstra2012} studies
the model with two stochastic interest rates.

Stochastic volatility models embody a spot volatility or variance process that
is correlated to the underlying asset. Popular choices are the Heston model
\cite{Heston1993} with variance following a Cox-Ingersoll-Ross (CIR) process
\cite{CIR1985}, and the Stein-Stein model \cite{SteinStein1991,
SchoebelZhu1999} with volatility following an arithmetic Ornstein-Uhlenbeck process
\cite{UhlenbeckOrnstein1930}. Both models admit closed form solutions, at least
in their constant parameter formulations, for pricing European vanilla options.
\cite{Atlan2006} proposes an extension to stochastic volatility models with a
stochastic interest rate. A more common extension is to join
the local volatility with the stochastic volatility or variance in the diffusion
term of the stochastic differential equation (SDE) to form a stochastic local
volatility model. In the literature, the theoretical framework for stochastic
local volatility has been studied in several forms with varying degrees of model
complexity. The ``unified theory of volatility'' from \cite{Dupire1996} imposes
dynamics on local variances, whereas \cite{AlexanderNogueira2004} focus on
modeling the stochastic evolution of the local volatility surface rather than of
the spot volatility. \cite{Lipton2002} proposes a
further extension to incorporate jumps in the model. Calibration
algorithms for several forms of stochastic local volatility models have been
introduced. \cite{Jex1999} implement a trinomial tree whereas, \cite{GuyonHL2013}
utilize particle methods on McKean nonlinear SDEs on general volatility models,
\cite{RenMadanQian2007, SaporitoYangZubelli2017} solve a forward Kolmogorov
PDE. The bootstrapping Monte Carlo method proposed by \cite{Labordere2009}
simplifies the model by Markovian projection. Our approach of simulating the
full stochastic local volatility model in our work paves the way to further
extend the model to incorporate stochastic rates.

We use a number of acronyms for the models under consideration for easy
reference. The standard Dupire local volatility model where the drift is the
difference between two deterministic rates is denoted by LV2DR. The stochastic
rates extension of this model is LV2SR, whereas the stochastic local volatility
extension is SLV2DR. The full model with stochastic rates and stochastic local
volatility is referred to as SLV2SR.

While this paper reviews and makes use of existing findings and algorithms in
part, to our knowledge, the following are our novel contributions:
SLV2SR model; complete derivations of the SDE systems that are used during Monte
Carlo calibration for all three models: LV2SR, SLV2DR, SLV2SR; presentation of
Monte Carlo calibration algorithms for the SLV2DR and SLV2SR models; and
convergence and accuracy studies of all calibration algorithms presented.

The structure of this paper is as follows. We end this section by reviewing
the theoretical foundations that our models will be based on.
In Section \ref{sec:LV2SR}, we give a detailed description of the algorithm to
calibrate the local volatility model subject to stochastic interest rates
(LV2SR), including derivations of the relevant formulas and SDEs in forward
measure. Section \ref{sec:SLV2DR} describes the
two-fold calibration algorithm of the stochastic local volatility model
(SLV2DR):
First the underlying Heston model is calibrated to match near at-the-money forward
(ATMF) market quotes; then the leverage function is calibrated to the entire
implied volatility surface. Section \ref{sec:SLV2SR} blends the algorithms from
the previous sections to develop a calibration algorithm for our ultimate model:
stochastic local volatility with stochastic interest rates (SLV2SR). As the
models being discussed are built on top of each other, the
calibration algorithm for the SLV2SR model references and uses the
results of the algorithms for calibrating the simpler models in earlier
sections. In this paper stochastic interest rates are modeled by Gaussian
one-factor model with deterministic shift (G1++) short rate processes. Other
short rate models can be adapted to our framework as well, as our methodology
is not tightly coupled with the
choice for the type of the short rate process. Each of these sections are
accompanied by tests that measure the convergence and the accuracy of the
calibrated models. We use EURUSD market data as of 2020-04-30 in our calibration
and simulation tests.
In Section \ref{sec:conclusion} we summarize our findings and
exhibit a study to compare the calibrated models.
\ref{sec:intrsds} proves a few technical lemmas and provides some computations
used in the paper.

We work in the standard probability space setting that can be found in most
mathematical finance textbooks. Here we briefly discuss the relevant constructs
for the sake of clarity, and for detail we refer to \cite{Shreve2004}.
We consider the measurable space $(\Omega, \mathcal{F})$, where
the sample space $\Omega$ with outcome elements $\omega$ governs the
uncertainty; $\mathcal{F}$, as a $\sigma$-algebra on $\Omega$, controls
information availability. Together with a probability measure $\mathbb{P}$, they
form the probability space $(\Omega, \mathcal{F}, \mathbb{P})$. The expectation
operator for the measure $\mathbb{P}$ is denoted by $\mathbf{E}^{\mathbb{P}}$.
The filtration $\{\mathcal{F}_t, t \geq 0 \}$, 
defined as a set over subalgebras $\mathcal{F}_t$ of $\mathcal{F}$ satisfying
$\mathcal{F}_s \subseteq \mathcal{F}_t$ for $s \leq t$, models the information
arrivals at different times.

We consider the setup where the filtration 
$\{\mathcal{F}_t, t \geq 0 \}$
is generated by $M$ independent Brownian motions $(W^m_t),
m=1,\ldots,M$, under $\mathbb{P}$.
We can formulate a model with $\mathcal{F}_t$-measurable stochastic processes
$(X^n_t)$ that are driven by $(W^m_t)$
\begin{equation*}
dX^n_t = \mu^n(t, \omega) dt + \sum_{m=1}^M \sigma^{nm}(t, \omega) dW^m_t,
\ n=1, \ldots , N,
\end{equation*}
by means of $\mathcal{F}_t$-measurable drift $\mu^n(t, \omega)$ and diffusion
$\sigma^{nm}(t, \omega)$ functions that for all $t$ satisfy the
integrability conditions
\begin{equation*}
\begin{split}
\int_0^t |\mu^n(s, \omega)| ds <& \infty,\\
\int_0^t \sum_{m=1}^M |\sigma^{nm}(s, \omega)|^2 ds <& \infty,
\end{split}
\end{equation*}

almost surely.
The stochastic processes that we are primarily considering in this paper are
exchange rate, stochastic volatility/variance, domestic and foreign interest
rates.
Other tradable assets we consider are domestic and foreign money market
accounts, which as num\'eraires define the domestic and foreign risk neutral
measures $\mathbb{Q}^{\text{DRN}}$ and $\mathbb{Q}^{\text{FRN}}$ respectively;
and domestic zero coupon bond maturing at future time $T$, which as
num\'eraire defines the domestic $T$-forward measure $\mathbb{Q}^{\text{T}}$.

The discounted asset values are martingales under the
risk neutral measure they are denominated in. It is convenient to formulate the
stochastic processes we will study in risk neutral measures, in which they
can be written in their canonical forms, driven by a single Brownian motion.
For example, a process $(X^n_t)$ is driven by the Brownian motion
$(W^{n\text{(DRN)}}_t)$ under the domestic risk neutral measure as
\begin{equation*}
dX^n_t = \tilde\mu^n(t, \omega) dt + \tilde\sigma^{n}(t, \omega)
dW^{n\text{(DRN)}}_t.
\end{equation*}
\begin{table}[ht!]
\begin{center}
\caption{The processes and the corresponding correlated Brownian motions under
the domestic risk-neutral measure}\label{tbl:allmodels}
\begin{tabular}{lcll}
 \hline \\[-12pt]
  & \multicolumn{1}{c}{$N\ (=M)$} &
 \multicolumn{1}{c}{$\{X^n_t\}$} &
 \multicolumn{1}{c}{$\{W^{n\text{(DRN)}}_t\}$}\\
 \hline & \\[-12pt]
LV2DR & 1 & $S_t$ & $W^{S\text{(DRN)}}_t$\\
LV2SR & 3 & $S_t, r^d_t, r^f_t$ & $W^{S\text{(DRN)}}_t,
W^{d\text{(DRN)}}_t, W^{f\text{(DRN)}}_t$\\
SLV2DR & 2 & $S_t, U_t$ & $W^{S\text{(DRN)}}_t, W^{U\text{(DRN)}}_t$\\
SLV2SR & 4 & $S_t, U_t, r^d_t, r^f_t$ & $W^{S\text{(DRN)}}_t,
W^{U\text{(DRN)}}_t, W^{d\text{(DRN)}}_t, W^{f\text{(DRN)}}_t$\\
 \hline
\end{tabular}
\end{center}
\end{table}
In this formulation the Brownian motions $(W^{n\text{(DRN)}}_t)$ of each process
$(X^n_t)$ are not necessarily uncorrelated. Note that any model $(X^n_t), n=1,
\ldots, N$, with correlated Brownian motions can be rewritten as a model with
uncorrelated Brownian motions, absorbing any correlation into the diffusion
terms. In this paper we specify the stochastic processes for the exchange
rate $(S_t)$ under various models. LV2SR and SLV2SR models specify the domestic
short rate $(r^d_t)$ and the foreign short rate
$(r^f_t)$ processes; and SLV2DR and SLV2SR models further specify the exchange
rate variance process $(U_t)$. Under domestic risk neutral measure
$\mathbb{Q}^{\text{DRN}}$, these processes are driven by the correlated Brownian
motions $(W^{S\text{(DRN)}}_t)$, $(W^{d\text{(DRN)}}_t)$,
$(W^{f\text{(DRN)}}_t)$, and $(W^{U\text{(DRN)}}_t)$, respectively. Table
\ref{tbl:allmodels} summarizes the models under consideration. A sample implementation of
the calibration algorithms presented in the following sections can be found at
the Github repository \url{https://github.com/oge-t/subtle_smile}.

\section{Local Volatility Model with Stochastic Rates (LV2SR)}\label{sec:LV2SR}
\subsection{Setup}

Let $(S_t)$ be the exchange rate, that is the amount of domestic currency needed
to buy one unit of foreign currency. In the domestic risk neutral measure
$\mathbb{Q}^{\text{DRN}}$ the exchange rate is assumed to follow the LV2SR local
volatility model,
\begin{equation}
dS_t = \left[r^d_t - r^f_t \right] S_t dt + \sigma(S_t, t) S_t
dW^{S\text{(DRN)}}_t,
\label{eqn:dS_LV_DRN}
\end{equation}
where $\sigma(S_t, t) > 0$ is the state dependent diffusion coefficient that is
commonly referred to as \emph{local volatility}, and $(S_t)$ is denominated in
domestic currency. The domestic short rate $(r^d_t)$ and the foreign short rate
$(r^f_t)$ follow G1++ processes. In particular, the domestic short rate
evolves in domestic risk neutral measure as
\begin{equation}
\begin{split}
r^d_t &= x^d_t + \phi^d_t, \\
dx^d_t &= -a^d_t x^d_t dt + \sigma^d_t dW^{d\text{(DRN)}}_t,
\label{eqn:drd_G1PP_DRN}
\end{split}
\end{equation}
whereas the foreign short rate evolves in foreign risk neutral measure
$\mathbb{Q}^{\text{FRN}}$ as
\begin{equation}
\begin{split}
r^f_t &= x^f_t + \phi^f_t, \\
dx^f_t &= -a^f_t x^f_t dt + \sigma^f_t dW^{f\text{(FRN)}}_t. \label{eqn:dx_fFRN}
\end{split}
\end{equation}
Here $\phi^i_t$ are the shift functions that are calibrated to market yield
curves; $a^i_t \geq 0$ are the mean reversion coefficients, and $\sigma^i_t > 0$
are the volatility coefficients, with $i=d, f$.

Our derivations and computations assume constant coefficients of correlation
between the returns of the underlying assets; however we note that it is
straightforward to generalize our findings to more advanced models with
time-dependent or even stochastic coefficients of correlation.

\subsubsection{Domestic risk neutral measure}
The first step is to derive the evolution of the foreign short rate in the
domestic risk neutral measure. For now assume that the evolution of the foreign
short rate in domestic risk neutral measure has the form
\begin{equation}
dx^f_t = g(\cdot, t) dt + \sigma^f_t dW^{f\text{(DRN)}}_t
\label{eqn:dx_fDRNimplicit}
\end{equation}
for some drift function $g(\cdot, t)$ of the underlying assets of the SDE
system and time that we are going to determine.

For any asset $V_t$ denominated in domestic
currency, the discounted asset price is a martingale under the domestic risk
neutral measure. Defining the domestic money market account as $B_t^d =
\exp[\int_0^t r_u^d du]$, we have
\begin{equation*}
\frac{V_0}{B_0^d} = V_0
= \mathbf{E}^{\mathbb{Q}^{\text{DRN}}}\left[\frac{V_t}{B_t^d}\right].
\end{equation*}
Likewise, the discounted value of $\frac{V_t}{S_t}$, that is the price of the
asset $V_t$ denominated in the foreign currency, is a martingale under the
foreign risk neutral measure. Defining the foreign money market account as
$B_t^f = \exp[\int_0^t r_u^f du]$, we have
\begin{equation*}
\frac{V_0}{S_0 B_0^f} = \frac{V_0}{S_0}
= \mathbf{E}^{\mathbb{Q}^{\text{FRN}}}\left[\frac{V_t}{S_t B_t^f}\right].
\end{equation*}
Therefore the Radon-Nikodym derivative \cite{Shreve2004} writes
\begin{equation}
\frac{d\mathbb{Q}^{\text{FRN}}}{d\mathbb{Q}^{\text{DRN}}} =
\frac{\frac{V_t}{B^d_t} \frac{V_0}{S_0}}{V_0 \frac{V_t}{S_t B^f_t}} =
\frac{S_t B^f_t}{S_0 B^d_t}. \label{eqn:RN1}
\end{equation}
The exchange rate process (\ref{eqn:dS_LV_DRN}) is an extension of the geometric
Brownian motion SDE with time-dependent coefficients $r^d_t$, $r^f_t$, and
time- and space-dependent coefficient $\sigma(S_t, t)$, with the solution that
describes its evolution from its spot value $S_0$,
\begin{equation*}
\begin{split}
S_t =& S_0 \exp\left[ \int_0^t
\left( r^d_u - r^f_u - \frac{\sigma^2(S_u, u)}{2} \right) du +
\int_0^t \sigma(S_u, u) dW^{S\text{(DRN)}}_u\right]\\
=&S_0 \frac{B^d_t}{B^f_t} \exp\left[ -\frac{1}{2}\int_0^t
\sigma^2(S_u, u) du +
\int_0^t \sigma(S_u, u) dW^{S\text{(DRN)}}_u\right].
\end{split}
\end{equation*}
Thus, the Radon-Nikodym derivative (\ref{eqn:RN1}) becomes
\begin{equation}
\frac{d\mathbb{Q}^{\text{FRN}}}{d\mathbb{Q}^{\text{DRN}}} =
\exp\left[ -\frac{1}{2}\int_0^t
\sigma^2(S_u, u) du +
\int_0^t \sigma(S_u, u) dW^{S\text{(DRN)}}_u\right].
\end{equation}
Then, according to the Girsanov theorem
\begin{equation}
dW^{S\text{(FRN)}}_t = dW^{S\text{(DRN)}}_t - \sigma(S_t, t) dt
\label{eqn:dW_SFRN_SDRN}
\end{equation}
is a Brownian motion under the foreign risk neutral measure
$\mathbb{Q}^{\text{FRN}}$.

Let $\rho_{Sf}$ be the coefficient of correlation between the Brownian motions
$(W^{S\text{(FRN)}}_t)$ and $(W^{f\text{(FRN)}}_t)$, that is
$
d\left<W^{S\text{(FRN)}}, W^{f\text{(FRN)}} \right>_t = \rho_{Sf} dt.
$
Following Lemma \ref{lemma:measure_trf}, the foreign short rate process
(\ref{eqn:dx_fDRNimplicit}) evolves in foreign risk neutral measure
$\mathbb{Q}^{\text{FRN}}$ as
\begin{equation}
dx^f_t = \left[ g(\cdot, t) + \rho_{Sf} \sigma^f_t \sigma(S_t, t) \right] dt +
\sigma^f_t dW^{f\text{(FRN)}}_t.\label{eqn:dx_fFRNimplicit}
\end{equation}
Comparing (\ref{eqn:dx_fFRN}) and (\ref{eqn:dx_fFRNimplicit}), we find that
\begin{equation*}
g(\cdot, t) = -a^f_t x^f_t - \rho_{Sf} \sigma^f_t \sigma(S_t, t).
\end{equation*}

Collectively, we can write the three processes in the domestic risk neutral
measure as
\begin{equation}
\begin{split}
dS_t =& \left[r^d_t - r^f_t \right] S_t dt + \sigma(S_t, t) S_t
dW^{S\text{(DRN)}}_t,\\
dx^d_t =& -a^d_t x^d_t dt + \sigma^d_t dW^{d\text{(DRN)}}_t,\ r^d_t = x^d_t +
\phi^d_t,\\
dx^f_t =& \left[ -a^f_t x^f_t - \rho_{Sf} \sigma^f_t \sigma(S_t, t) \right] dt +
\sigma^f_t dW^{f\text{(DRN)}}_t,\ r^f_t = x^f_t + \phi^f_t. \label{eqn:SDEs_DRN}
\end{split}
\end{equation}

\subsubsection{Domestic $T$-forward measure}\label{sec:Tforward}

For computational ease, that is to decouple the discounting terms from
expectations, we would like to transform (\ref{eqn:SDEs_DRN}) to the
domestic $T$-forward measure. We take the zero coupon bond $P^d(t, T)$ maturing
at time $T$ as the num\'{e}raire. We have $P^d(T, T) = 1$. Under the measure
$\mathbb{Q}^{\text{T}}$ defined by this num\'{e}raire, the discounted price of
an asset is a martingale,
\begin{equation*}
\frac{V_0}{P^d(0, T)}
= \mathbf{E}^{\mathbb{Q}^{\text{T}}}\left[\frac{V_T}{P^d(T, T)}\right]
= \mathbf{E}^{\mathbb{Q}^{\text{T}}}\left[V_T\right].
\end{equation*}
We arrive at the following Radon-Nikodyn derivative,
\begin{equation}
\frac{d\mathbb{Q}^{\text{T}}}{d\mathbb{Q}^{\text{DRN}}} =
\frac{\frac{V_T}{B^d_T} \frac{V_0}{P^d(0, T)}}{V_0 V_T} =
\frac{1}{B^d_T P^d(0, T)} =
\frac{\exp{\left[ - \int_0^T r^d_u du\right]}}{P^d(0, T)}.
\end{equation}
Since the domestic short rate follows a G1++ process, this expression can be
written as (see Lemma \ref{lemma:g1pp_identity} for proof)
\begin{equation}
\frac{d\mathbb{Q}^{\text{T}}}{d\mathbb{Q}^{\text{DRN}}} =
\exp\left[- \int_0^T \sigma^d_u b^d(u, T) dW^{d\text{(DRN)}}_u
-\frac{1}{2} \int_0^T \left(\sigma^d_u b^d(u, T)\right)^2 du
\right],\label{eqn:RN_dT_dDRN}
\end{equation}
with
\begin{equation*}
b^d(t, T) \equiv \int_t^T e^{- \int_t^v a^d_z dz} dv.
\end{equation*}
Then, by Girsanov theorem
\begin{equation}
dW^{d\text{(T)}}_t = dW^{d\text{(DRN)}}_t + b^d(t, T) \sigma^d_t dt
\label{eqn:dW_dT_dDRN}
\end{equation}
is a Brownian motion under the domestic $T$-forward measure
$\mathbb{Q}^{\text{T}}$.
This allows us to write down the domestic short rate process
(\ref{eqn:drd_G1PP_DRN}) as
\begin{equation}
dx^d_t = \left[-a^d_t x^d_t - b^d(t, T) (\sigma^d_t)^2 \right]
dt + \sigma^d_t dW^{d\text{(T)}}_t.
\end{equation}

Let $\rho_{Sd}$ be the coefficient of correlation between the Brownian motions
$(W^{S\text{(DRN)}}_t)$ and $(W^{d\text{(DRN)}}_t)$, that is
$
d\left<W^{S\text{(DRN)}}, W^{d\text{(DRN)}} \right>_t = \rho_{Sd} dt.
$
Following Lemma \ref{lemma:measure_trf}, the exchange rate process
(\ref{eqn:dS_LV_DRN}) evolves in domestic $T$-forward measure
$\mathbb{Q}^{\text{T}}$ as
\begin{equation}
dS_t = \left[ r^d_t - r^f_t - \rho_{Sd} b^d(t, T) \sigma^d_t \sigma(S_t, t)
\right] S_t dt + \sigma(S_t, t) S_t dW^{S\text{(T)}}_t.\label{eqn:LV2SR_ST}
\end{equation}

Finally, let $\rho_{df}$ be the coefficient of correlation between the Brownian
motions $(W^{d\text{(DRN)}}_t)$ and $(W^{f\text{(DRN)}}_t)$, that is
$
d\left<W^{d\text{(DRN)}}, W^{f\text{(DRN)}} \right>_t = \rho_{df} dt.
$
Following Lemma \ref{lemma:measure_trf}, the foreign short rate process
from (\ref{eqn:SDEs_DRN}) evolves in domestic $T$-forward measure
$\mathbb{Q}^{\text{T}}$ as
\begin{equation}
dx^f_t = \left[ -a^f_t x^f_t - \rho_{Sf} \sigma^f_t \sigma(S_t, t) 
- \rho_{df} b^d(t, T) \sigma^d_t \sigma^f_t \right] dt +
\sigma^f_t dW^{f\text{(T)}}_t.
\end{equation}

Collecting everything,
\begin{equation}
\begin{split}
dS_t =& \left[ r^d_t - r^f_t - \rho_{Sd} b^d(t, T) \sigma^d_t \sigma(S_t, t)
\right] S_t dt + \sigma(S_t, t) S_t dW^{S\text{(T)}}_t,\\
dx^d_t =& \left[-a^d_t x^d_t - b^d(t, T) (\sigma^d_t)^2 \right] dt + \sigma^d_t
dW^{d\text{(T)}}_t,\ r^d_t = x^d_t + \phi^d_t,\\
dx^f_t =& \left[ -a^f_t x^f_t - \rho_{Sf} \sigma^f_t \sigma(S_t, t) 
- \rho_{df} b^d(t, T) \sigma^d_t \sigma^f_t \right] dt +
\sigma^f_t dW^{f\text{(T)}}_t, \ r^f_t = x^f_t + \phi^f_t\label{eqn:SDEs_T}
\end{split}
\end{equation}
describe the evolutions of the exchange rate, domestic short rate, and foreign
short rate processes under the domestic $T$-forward measure
$\mathbb{Q}^{\text{T}}$. Note that the above SDE system is different than what
is given in \cite{Deelstra2012}.

\subsection{Calibration of Local Volatility}\label{sec:calib_lv_sr}

The standard formulation of the local volatility model \cite{Dupire1994,
DermanKani1994} with deterministic interest rates has been studied extensively
in the literature. The local volatility surface can be computed from a
call option price surface that can be constructed by market quotes of call
option prices $C = C(K, T)$ as \cite{Gatheral2012}
\begin{equation}
   \sigma_{\text{LV (deterministic rates)}}^2 = \frac{\frac{\partial C}{\partial
   T} + (r^d_T - r^f_T) K \frac{\partial C}{\partial K} +
   r^f_T C}{\frac{1}{2}K^2 \frac{\partial^2C}{\partial K^2}}.
   \label{eqn:dupire_C_deterministic_ir}
\end{equation}
When the interest rates have stochastic dynamics, the above equation generalizes
to \cite{Ogetbil2020}
\begin{equation}
   \sigma_{\text{LV (stochastic rates)}}^2 = \frac{\frac{\partial C}{\partial T}
   - P^d(0, T) \mathbf{E}^{\mathbb{Q}^{\text{T}}}\left[(K r_T^d - S_T r_T^f) \mathds{1}_{S_T > K}\right]}
   {\frac{1}{2}K^2 \frac{\partial^2C}{\partial K^2}}.
   \label{eqn:dupire_C_stochastic_ir}
\end{equation}
Here $P^d(0, T)$ is the time zero value of a zero coupon bond expiring at time
$T$, which can be extracted from the input domestic discount factor curve. The
expectation is taken under the domestic $T$-forward measure
$\mathbb{Q}^{\text{T}}$.

One can show that when the rates $(r^d_t)$, $(r^f_t)$ are deterministic
($\sigma^d_t \rightarrow 0, \sigma^f_t \rightarrow 0$),
(\ref{eqn:dupire_C_stochastic_ir}) reduces to
(\ref{eqn:dupire_C_deterministic_ir}). In the stochastic case,
however, this expectation does not have a known analytical solution and needs to
be evaluated numerically. Below, we demonstrate our methods for evaluating this
expectation by Monte Carlo simulation.

The evaluation of the local volatility requires the construction of the
\emph{call price surface} interpolator in this formulation. The interpolator
must be able to evaluate the partial derivatives appearing in the above local
volatility expressions\footnote{In our implementation, the interpolator is
constructed as clamped cubic spline in strike direction, and linear spline
in time direction. The derivatives are computed analytically on the
splines. Expiries are uniformly spaced. At each expiry the strikes range 3.5
standard deviations (based on the corresponding implied volatility) away from
at-the-money-forward strike, and are uniformly spaced in log-moneyness.}.
This surface and its $t$- and $K$-derivatives will be evaluated at a given grid
to generate the local volatility surface using equations
(\ref{eqn:dupire_C_deterministic_ir}) or (\ref{eqn:dupire_C_stochastic_ir}).

As the FX volatility market data is usually given in terms of risk-reversals and
butterflies, and represented as some \emph{implied volatility surface} form
with appropriate interpolation, we choose to write the Dupire's equations in
the total implied variance parametrization.
Comparing the calibration routines on various sets of initial market data, we
found that in the wings of the surface, the total implied variance formulation
typically performs better than the call surface formulation. In
practice, market data is usually available in the form of parametrized or dense
implied volatility surfaces that are calibrated with such penalty functions that
force the total implied variance to increase monotonically as a function of
time. By construction, interpolating the total implied variance surface and
using these values in the Dupire formula avoids calendar spread arbitrage.

Following the conventions of \cite{Gatheral2012},
the total implied variance $w$ is parametrized in terms of the log-moneyness
$y(K, T) = \log\frac{K}{F_T}$ and time $T$, with
forward asset price $F_T = S_0 \frac{P^f(0, T)}{P^d(0, T)}$, as $w(y(K,T),T) =
\Sigma(K,T)^2 T$ where $\Sigma(K,T)$ is the market implied volatility at strike
$K$ and maturity $T$.

In the deterministic interest rates case, the Dupire's equation
(\ref{eqn:dupire_C_deterministic_ir}) reduces to \cite{Gatheral2012}
\begin{equation}
   \sigma_{\text{LV (deterministic rates)}}^2 = \frac{\frac{\partial w}{\partial
   T}} {1 - \frac{y}{w}\frac{\partial w}{\partial y} + \frac{1}{2} \frac{\partial^2
   w}{\partial y^2} + \frac{1}{4} \left(\frac{\partial w}{\partial y}\right)^2
   \left(-\frac{1}{4} - \frac{1}{w} + \frac{y^2}{w^2}\right)}.
   \label{eqn:dupire_tiv_deterministic_ir}
\end{equation}
Meanwhile, in the stochastic interest rates case, the equation
(\ref{eqn:dupire_C_stochastic_ir}) becomes \cite{Ogetbil2020}
\begin{equation}
   \sigma_{\text{LV (stochastic rates)}}^2 = \frac{
\frac{\partial C_{\text{BS}}}{\partial T}
- P^d(0, T) \mathbf{E}^{\mathbb{Q}^{\text{T}}}\left[(K r_T^d - S_T r_T^f) \mathds{1}_{S_T > K}\right]}
{\frac{\partial C_{\text{BS}}}{\partial w} \left[
1 - \frac{y}{w} \frac{\partial w}{\partial y} + \frac{1}{2} \frac{\partial^2 w}{\partial y^2}
+\frac{1}{4} \left(\frac{\partial w}{\partial y}\right)^2
\left(-\frac{1}{4}- \frac{1}{w} + \frac{y^2}{w^2}\right)
\right]},\label{eqn:dupire_tiv_stochastic_ir}
\end{equation}
where the Black-Scholes model price $C_{\text{BS}}=C_{\text{BS}}(P^d(0, T) F_T, y, w)$ and
its derivatives are given by
\begin{equation}
\begin{split}
C_{\text{BS}}(P^d(0, T) F_T, y, w) =& P^d(0, T) F_T \left[N(d_1) - e^y N(d_2)\right],\\
   \frac{\partial C_{\text{BS}}}{\partial w} =& \frac{1}{2} P^d(0, T) F_T e^y
 N'(d_2) w^{-\frac{1}{2}},\\
   \frac{\partial C_{\text{BS}}}{\partial y} =& -P^d(0, T) F_T e^y N(d_2),\\
   \frac{\partial C_{\text{BS}}}{\partial T} =& - f^f(0, T) C_{\text{BS}}
+ \frac{\partial C_{\text{BS}}}{\partial w} \frac{\partial w}{\partial T}\\
&+\left(\frac{\partial C_{\text{BS}}}{\partial y}
+ \frac{\partial C_{\text{BS}}}{\partial w} \frac{\partial w}{\partial y}\right) (f^f(0, T)
- f^d(0, T)).
\end{split}\label{eqn:CBS_first_derivatives}
\end{equation}
Here the instantaneous forward rate is defined as
$f^i(0, T) \equiv - \frac{\partial \log P^i(0, T) }{\partial T}
= - \frac{1}{P^i(0, T) }\frac{\partial P^i(0, T) }{\partial T}$, with $i=d,f$.
Moreover $N(\cdot)$ is the cumulative Gaussian probability distribution
function; $d_1 = -y w^{-\frac{1}{2}} + \frac{1}{2} w^{\frac{1}{2}},$ and $d_2 =
d_1 - w^{\frac{1}{2}}$. $r_T^d$ and $r_T^f$ are the time $T$ values of the
domestic and foreign short rates.
Similar to the call price surface formulation, equation
(\ref{eqn:dupire_tiv_stochastic_ir}) reduces to
(\ref{eqn:dupire_tiv_deterministic_ir}) in the deterministic interest rates
limit. Analogous to the call price formulation case, the evaluation of the local
volatility requires the construction of the \emph{total implied variance
surface}.

Dupire's equations (\ref{eqn:dupire_C_stochastic_ir}) and
(\ref{eqn:dupire_tiv_stochastic_ir}) use the local volatility on both sides
since the computation of the expectation on the right hand side is under the
dynamics that involves the local volatility function. This can be used to define
iterative approaches. We found that the bootstrapping approach presented below
typically yields satisfactory results without any iterative refinement.

\paragraph{Inputs for calibration}

Our calibration routine expects the following quantities as input for local
volatility calibration:
\begin{itemize}
\item Spot FX rate $S_0$
\item Market implied volatility $\Sigma(K, t)$ for FX rate
\item Market yield curves $P^d(0, t)$ and $P^f(0, t)$
\item For both domestic and foreign rates, G1++ model parameters mean
reversion, volatility and shift function calibrated to market data\footnote{See
\cite{GurrieriNakabayashiWong2009} for example calibration methods for both
constant and time dependent cases.}
\item Coefficients of correlation between the underlying assets: the FX rate,
the domestic and foreign short rates
\end{itemize}

\paragraph{Steps for calibration}

In our framework, we calibrate the local volatility surface time slice by
time slice, in a bootstrapping fashion. Let $t_i; i=1, \ldots, n$ be the
increasing sequence of (positive) times where we will perform the calibration.
\begin{enumerate}
  \item Using the market implied volatility $\Sigma(K, t)$, generate a vanilla
  call option price surface $C(K, t)$ interpolator or a total implied variance
  surface $w(y, t)$ interpolator. The interpolator must be able to compute the
  partial derivatives appearing in the local volatility expressions.
  \item For the first time slice $t_1$, evaluate the
  deterministic equation (\ref{eqn:dupire_C_deterministic_ir}) or
  (\ref{eqn:dupire_tiv_deterministic_ir}) to compute the FX local volatilities
  for a predetermined range of strikes. This step requires no Monte Carlo
  simulation. As a result, obtain local volatility values to be used in
  the simulation until time $t_2$ in the subsequent calibration steps.
  \item For each of the subsequent time slices $t_j, j > 1$, simulate the SDE
  system (\ref{eqn:SDEs_T}) up to time $t_j$. Compute the Monte Carlo
  estimate for the expectation appearing in (\ref{eqn:dupire_C_stochastic_ir}) or
  (\ref{eqn:dupire_tiv_stochastic_ir}) for a predetermined range of strikes.
  Use these equations to obtain the local volatility values. These local
  volatility values will be used during subsequent simulation steps from time
  $t_j$ to time $t_{j + 1}$. This step is first performed with $j=2$ and is
  then repeated for the remaining time slices.
\end{enumerate}
The strike grid can be chosen to be uniform across all calibration time slices.
In this approach, however, the strike grid needs to be sufficiently large to
cover attainable values of the FX rate at long expiries. This in turn would
result in unreliable local volatility values at short expiries and strikes in
the far wings. To overcome this problem, we suggest using a more adequate
strike grid at each calibration time slice, e.g. one that spans a predetermined
number of standard deviations away from the ATMF strike value. Another plausible
approach is choosing the strike grid to cover the range of strikes that the
implied volatility surface is calibrated to, if this information is available.

\subsection{Calibration and Simulation Tests}

In order to test the validity of the calibrated local volatility surface, one
needs to use it for pricing. The prices generated by Monte Carlo simulation by
this method, however, have two sources of Monte Carlo errors. First, the
estimation of the expectations appearing in (\ref{eqn:dupire_C_stochastic_ir}) or
(\ref{eqn:dupire_tiv_stochastic_ir}) is subject to Monte Carlo error. Second,
the evaluation of the Monte Carlo average of the payoffs computed during the
pricing introduces an additional source of error. Keeping the number of paths
high in one of these two steps will allow us to study the convergence of the
other.

For both calibration and pricing, we also generate the antithetic conjugate
paths to reduce the variance. Therefore when we talk about a simulation with $N$
as the number of paths, the actual total number of paths simulated is $2N$.

The G1++ model
parameters we use as input are summarized in Table \ref{table:g1pp_parameters}.
The coefficients of correlation are given by $\rho_{Sd} = 0.166, \rho_{Sf} =
0.551, \rho_{df} = 0.161$.
The domestic $T$-forward measure SDE system
(\ref{eqn:SDEs_T}) is simulated via forward Euler discretization in both calibration and pricing steps. We vary the number of
paths used during calibration while we fix the number of paths for the pricing
simulation at 100,000.
\begin{table}[ht!]
\footnotesize
\begin{center}
\caption{\footnotesize G1++ model parameters for the domestic currency USD and
foreign currency EUR used as input data in our calibration
routine. In our implementation these are
interpolated piecewise constantly with left
endpoints.}\label{table:g1pp_parameters}
\begin{tabular}{lrrrr}
 \hline
 \multicolumn{1}{c}{$t$} & \multicolumn{1}{c}{$\sigma^d_t$} &
 \multicolumn{1}{c}{$\sigma^f_t$} & \multicolumn{1}{c}{$a^d_t$} &
 \multicolumn{1}{c}{$a^f_t$} \\
 \hline
 0 & 0.00956 & 0.00820 & 0.02 & 0.02 \\
 0.24 & 0.00881 & 0.00777 \\
 0.50 & 0.00851 & 0.00754 \\
 0.99 & 0.00804 & 0.00725 \\
 1.99 & 0.00818 & 0.00713 \\
 2.99 & 0.00816 & 0.00747 \\
 4.99 & 0.00810 & 0.00740 \\
 9.99 & 0.00818 & 0.00771 \\
 14.99 & 0.00818 & 0.00771 \\
 \hline
\end{tabular}
\end{center}
\end{table}

In this test, we use the total implied variance formulation
(\ref{eqn:dupire_tiv_stochastic_ir}). The Monte Carlo estimate for the
local volatility is given by
\begin{equation*}
   \hat{\sigma}_{\text{LV}}^2 = \frac{
\frac{\partial C_{\text{BS}}}{\partial T}
- P^d(0, T) \hat{E}}
{\frac{\partial C_{\text{BS}}}{\partial w} \left[
1 - \frac{y}{w} \frac{\partial w}{\partial y} + \frac{1}{2} \frac{\partial^2 w}{\partial y^2}
+\frac{1}{4} \left(\frac{\partial w}{\partial y}\right)^2
\left(-\frac{1}{4}- \frac{1}{w} + \frac{y^2}{w^2}\right)
\right]},
\end{equation*}
where $\hat{E}$ is the Monte Carlo estimate for the expectation appearing in
(\ref{eqn:dupire_tiv_stochastic_ir}). In this formulation, the Monte Carlo error
$\delta \hat{E}$ of this estimate translates to the error in local
volatility as
\begin{equation*}
\delta \hat{\sigma}_{\text{LV}} = \left|
\frac{P^d(0, T) \delta\hat{E}}
{2 \hat{\sigma}_{\text{LV}}
\frac{\partial C_{\text{BS}}}{\partial w} \left[
1 - \frac{y}{w} \frac{\partial w}{\partial y} + \frac{1}{2} \frac{\partial^2 w}{\partial y^2}
+\frac{1}{4} \left(\frac{\partial w}{\partial y}\right)^2
\left(-\frac{1}{4}- \frac{1}{w} + \frac{y^2}{w^2}\right)
\right]}
\right|.
\end{equation*}
For every point of the calibrated surface, we estimate the
Monte Carlo error in the local volatility values. In the end we price with three
volatility surfaces: the original calibrated surface, the original calibrated
surface bumped down by 2 Monte Carlo errors, and the original calibrated
surface bumped up by 2 Monte Carlo errors.

The total implied variance interpolator
has a time slice every 0.05 years, and 100 log-moneynesses per slice. The
log-moneyness points are spanned uniformly over 3.5 standard deviations from the
ATMF strike value for each slice. Similarly, the local volatility
surface is calibrated to have a time slice every 0.05 years. The strike grid of
the local volatility surface has 200 points, spanned uniformly over 3 standard
deviations from the ATMF strike value. The maximum simulation time step
is set at 0.01 years for both calibration and pricing.

We price a set of vanilla call options expiring at 10 years. Since we have the
market implied volatility surface data $\Sigma(K, t)$, we can compare the Monte
Carlo prices to analytical Black-Scholes vanilla call option prices implied by
$\Sigma(K, t)$.

\begin{figure}[ht!]
\centering
\includegraphics[width=0.495\textwidth]{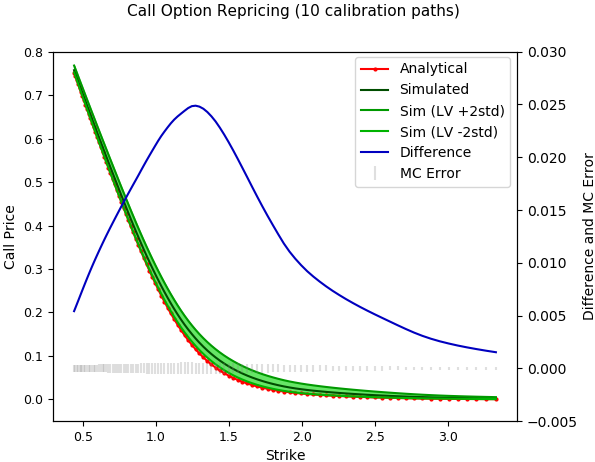}
\includegraphics[width=0.495\textwidth]{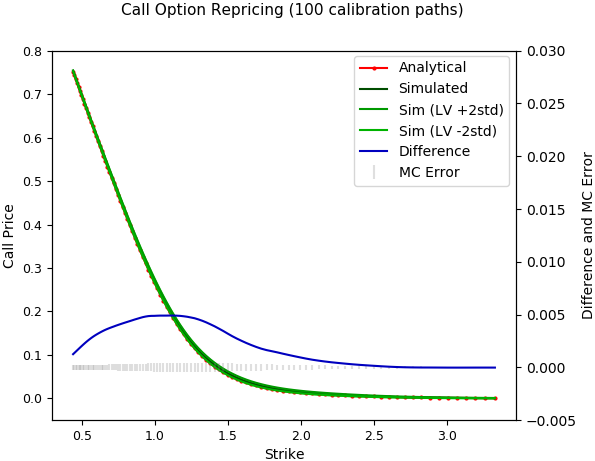}
\includegraphics[width=0.495\textwidth]{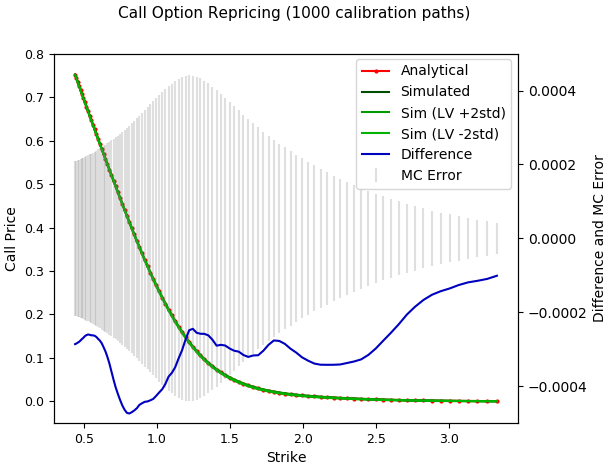}
\includegraphics[width=0.495\textwidth]{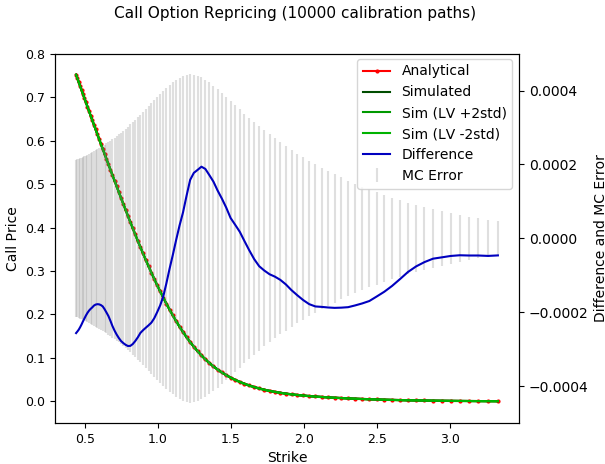}
\caption{LV2SR: Repricing of vanilla call options with local volatility surfaces
calibrated with varying number of paths. The pricing simulation is done with
100,000 paths to keep the simulation Monte Carlo error small. The pricing
is done with the original calibrated surface as well as bumped down and bumped
up surfaces according to Monte Carlo error introduced during the calibration.
Comparison to Black-Scholes prices is made by observing the differences between
the Monte Carlo and the analytical prices. One can observe that the pricing
differences become comparable in magnitude to the simulation Monte Carlo errors
at around 1,000 calibration paths.}
\label{fig:Cal_convergence} 
\end{figure}

In Figure \ref{fig:Cal_convergence} we see that the convergence is
achieved quickly with a relatively low number of paths at 1,000, where the
pricing differences begin to be comparable in magnitude to the simulation Monte
Carlo errors. The relatively low number of paths allow fast calibration.

\FloatBarrier


For the convergence of the calibrated local volatility surface, 
two surfaces calibrated with 100,000 and 50,000 calibration paths respectively, 
were studied by looking at the difference between the
surfaces in Figure \ref{LV2SR_convergence}. For the surfaces the x-axis
represents the spot values of the underlier, the y-axis the time and z-axis the
corresponding local volatility values. It can be seen that difference between
the surfaces is roughly 1\% of the magnitude of the local volatilities and the
convergence of the surface is achieved for relatively small number of paths.

\begin{figure}[ht!]
    \centering
    \includegraphics[width=\textwidth]{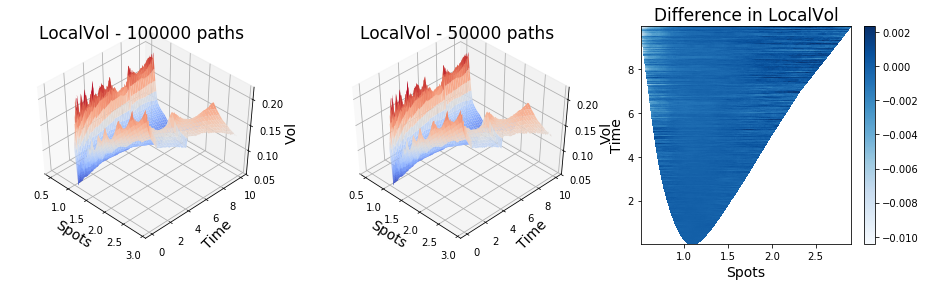}
    \caption{LV2SR: Calibrated local volatility surfaces with 100,000
    calibration paths (left), 50,000 calibration paths (middle) and
    difference between the surfaces (right).
    The difference between the surfaces is roughly 1\% of the magnitude of the
    local volatilities.}
    \label{LV2SR_convergence}
\end{figure}

\FloatBarrier


In order to study the effects of repricing, a local volatility
surface calibrated with 100,000 paths was used for all subsequent tests presented
in this section. The effect of number of repricing simulation paths for the
given local volatility surface is shown in Figure \ref{LV2SR_maturity_repricing}
against the call option price and the corresponding MC error. As the MC error
decreases with the number of simulation paths, between 1,000 and 100,000
paths the maximum absolute difference between the analytical and Monte Carlo priced call
option values decreases from $9.1\times 10^{-4}$ to $5.1\times 10^{-4}$.

\begin{figure}[ht!]
    \centering \includegraphics[width=\textwidth]{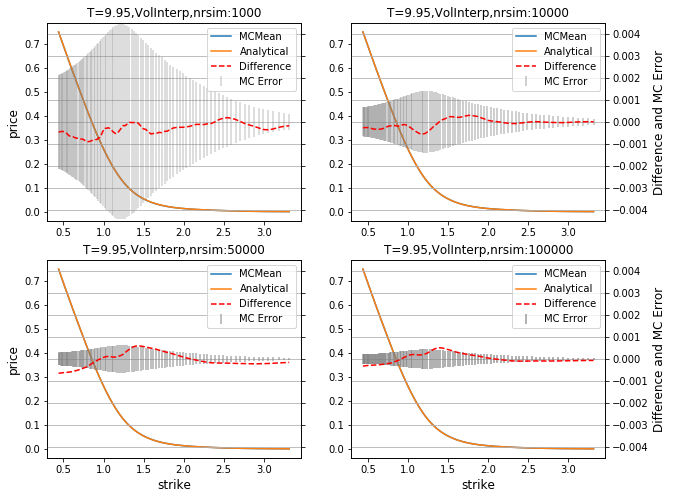}
    \caption{LV2SR: Repricing call options at 100 uniformly spaced strikes at
    maturity T=9.95 years, each with 100,000 (lower-right), 50,000 (lower-left), 10,000
    (upper-right) and 1,000 (upper-left) MC simulation paths,
    using the local volatility surface calibrated with 100,000 MC simulation
    paths. The MC errors and the difference between the analytical and MC
    computed prices decreases with the number of simulation paths.}
    \label{LV2SR_maturity_repricing}
\end{figure}

\FloatBarrier


Furthermore, the calibrated local volatility surface was used to reprice the
call options at multiple maturities and various strikes in the strike-grid to
generate the so-called call price surface. The 100 maturies that are uniformly
spaced between
T=0 and 9.95 years, and 100 strikes per maturity were used to generate the call
price surface shown in Figure \ref{fig:lv2srcallsurface}a. The difference
between the Monte Carlo repriced call option values and analytical Black-Scholes
call option prices assuming constant interest rates and volatilty is shown in
Figure \ref{fig:lv2srcallsurface}b, which is found to be less than 0.1\% of the
call option price.

\begin{figure}[ht!]
\begin{subfigure}{.49\textwidth}
  \centering
  \includegraphics[width=\linewidth]{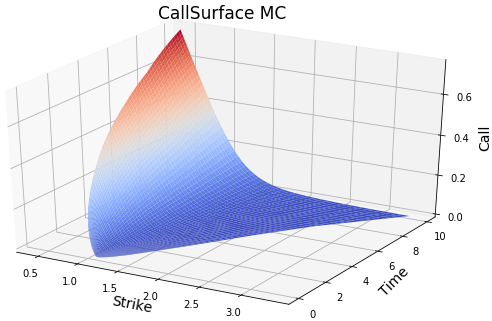}  
  \caption{Repriced call option surface at all strikes and maturities in the grid}
\end{subfigure}
\begin{subfigure}{.49\textwidth}
  \centering
  \includegraphics[width=\linewidth]{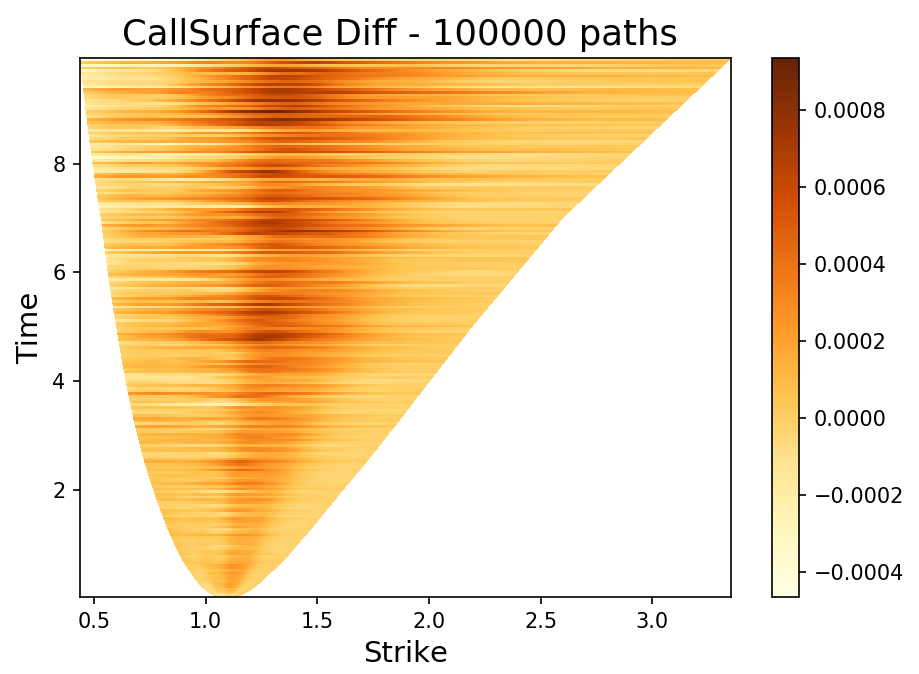}  
  \caption{Difference in analytical call surface and repriced call surface}
\end{subfigure}
\caption{LV2SR: (a) Call options repriced at 100 uniformly spaced maturities
between T=0 to 9.95 years and 100 strikes per maturity and (b) the difference
between Monte Carlo and Black-Scholes analytical price.}
\label{fig:lv2srcallsurface}
\end{figure}

\FloatBarrier


Next, the market implied volatility and the implied volatility
recovered from the Monte Carlo repriced options, by inverting the option price
to evaluate Black-Scholes implied volatility, are compared in Figure
\ref{LV2SR_maturity_repricing2} at maturity T=9.95. Additionally the implied
volatility recovered from MC prices of out-of-the-money call and put options
with $\pm 2$ MC errors is presented.

It can be seen that recovered implied volatility is in good agreement with market
implied volatility and the latter is found to be well within the $\pm 2$ MC error
bounds.

\begin{figure}[ht!]
    \centering \includegraphics[width=\textwidth]{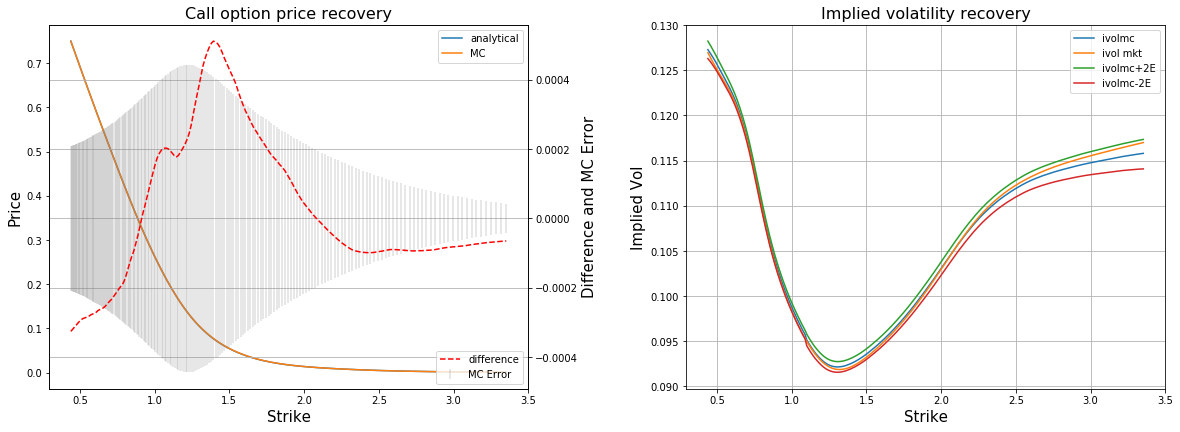}
    \caption{LV2SR: Difference between MC repriced vs analytical call
    options (left), implied volatility computed from repriced out-of-the-money
    call and put options vs market implied volatility at maturity(right) (T=9.95)}
    \label{LV2SR_maturity_repricing2}
\end{figure}

\FloatBarrier


Finally, this procedure is repeated for all the
maturities (time slices) in the repriced call surface, where the market and
recovered implied volatility along with implied volatilities corresponding to
$\pm 2$ MC pricing errors are shown for a few of the slices in Figure
\ref{LV2SR_ivol_recovery}.

\begin{figure}[ht!]
    \centering \includegraphics[width=\textwidth]{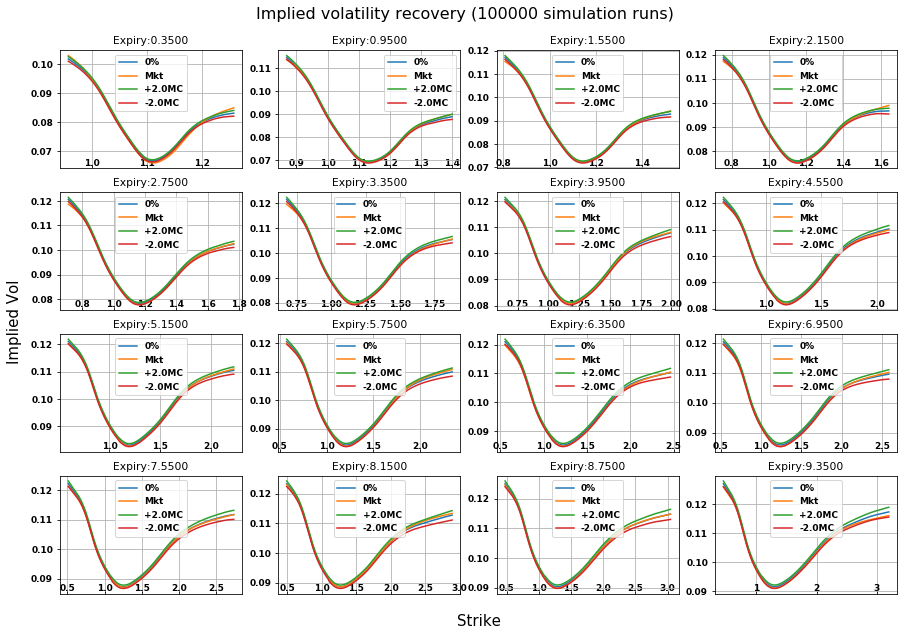}
    \vspace{-20pt}
    \caption{LV2SR: Implied volatility computed from repriced out-of-the-money
    call and put options vs the market implied volatility at various maturities}
    \label{LV2SR_ivol_recovery}
\end{figure}

\section{Stochastic Local Volatility Model (SLV2DR)} \label{sec:SLV2DR}
\subsection{Setup}

The standard local volatility model with deterministic interest rates (LV2DR)
\begin{equation}
   dS_t = \left[r^d_t - r^f_t\right] S_t dt + \sigma_{\text{LV}}(S_t, t) S_t
   dW_t^{S\text{(DRN)}}
   \label{eqn:sde_stdlv_domrn}
\end{equation}
can be extended to incorporate a stochastic nature in the diffusion term by
replacing $\sigma_{\text{LV}}(S_t, t)$ with $L(S_t, t) \sqrt{U_t}$ where
$L(S_t, t) > 0$ is the \emph{leverage function} and $(U_t)$ is the variance
process.
A common choice for $(U_t)$ is the Cox-Ingersoll-Ross (CIR) process
\cite{CIR1985}. With this choice the SLV2DR SDE system is
\begin{align}
\begin{split}
   dS_t =& \left[r^d_t - r^f_t\right] S_t dt + L(S_t, t) \sqrt{U_t} S_t
   dW_t^{S\text{(DRN)}},\\
   dU_t =& \kappa_t (\theta_t - U_t) dt + \xi_t \sqrt{U_t} dW_t^{U\text{(DRN)}},
\end{split}\label{eqn:SDEs_SLV}
\end{align}
with coefficient of correlation $\rho_{SU}$ between the two Brownian drivers.
This model can be seen as an augmentation of the Heston model
\cite{Heston1993}, with $\kappa_t > 0$, $\theta_t > 0$, and $\xi_t > 0$
representing the time-dependent mean reversion, long term variance, and
vol-of-vol parameters. Together with the
initial variance $U_0 > 0$, they form the set of Heston model parameters that
will be calibrated to market data as we will describe below. The leverage function
$L(S_t, t)$ is to be calibrated to recover market option prices.

The standard local volatility $\sigma_{\text{LV}}(S_t, t)$ is
related to the leverage function $L(S_t, t)$ as
\cite{Ogetbil2020}\footnote{We note that the expectation in
\cite{Ogetbil2020} is given in domestic $T$-forward measure
$\mathbb{Q}^{\text{T}}$, which is equal to domestic risk neutral measure
$\mathbb{Q}^{\text{DRN}}$ under deterministic rates.}
\begin{equation}
   \sigma_{\text{LV}}(x, t)^2 = L(x, t)^2
   \mathbf{E}^{\mathbb{Q}^{\text{DRN}}}\left[U_t \mid S_t=x \right].
   \label{eqn:relation_leverage_locvol}
\end{equation}

The main idea is to have the Heston parameters recover market vanilla option
prices near ATMF strikes. The leverage function will then serve as a
correction factor at the wings of the volatility surface.

\subsection{Calibration of the Heston model parameters} \label{sec:heston_calib}

In the limit where the calibration function is set to $L(S_t, t)=1$, the
stochastic local volatility model (\ref{eqn:SDEs_SLV}) reduces to the Heston
model with time dependent coefficients,
\begin{align}
\begin{split}
   dS_t =& \left[r^d_t - r^f_t \right] S_t dt + \sqrt{U_t} S_t
   dW_t^{S\text{(DRN)}},\\
   dU_t =& \kappa_t (\theta_t - U_t) dt + \xi_t \sqrt{U_t} dW_t^{U\text{(DRN)}}.
\end{split}\label{eqn:SDEs_Heston}
\end{align}
Here we assume constant correlation between the two Brownian motions,
\begin{equation*}
\left<dW^{S\text{(DRN)}}, dW^{U\text{(DRN)}}\right>_t = \rho dt.
\end{equation*}
To improve calibration
accuracy, one can trivially extend the model to admit time dependent
correlation. However, we found that our simpler setup is sufficient for our
purposes. The set of parameters we need to calibrate are the mean reversion
$\kappa_t > 0$, the long term variance $\theta_t > 0$, the vol-of-vol $\xi_t >
0$, the coefficient of correlation $\rho$, and the initial variance $U_0 > 0$.
Several methods have been studied in literature to calibrate these parameters,
including an asymptotic approximation \cite{MikhailovNoegel2005}, or
a semi-analytical approach computing the characteristic function and using
control variates to regularize the numerical integration
\cite{GuterdingBoenkost2018}. While these methods can be directly applied to
our setup, we choose a simpler approach of calibrating the parameters using a
PDE solver\footnote{The PDE solver implements the two-dimensional backward
Kolmogorov PDE for the Heston model with time dependent coefficients, using a
standard finite difference alternating direction implicit (ADI) scheme
\cite{Lin2008}.}.
We note that this is not a binding choice, and one can use any reasonably
fast solver to be repetitively called by the objective function of the
optimizer that can accurately price under this model.

\paragraph{Inputs for calibration}

Our calibration routine expects the following quantities as input for
Heston model with piecewise constant coefficients calibration:
\begin{itemize}
\item Spot FX rate $S_0$
\item Market implied volatility $\Sigma(K, t)$ for FX rate
\item Market yield curves $P^d(0, t)$ and $P^f(0, t)$. Since the rates are
deterministic, we have $r^i_t = f^i(0, t), i = d, f$ where the instantaneous
forward rate can be computed from the market yield curves, $f^i(0, t) \equiv -
\frac{\partial \log P^i(0, t) }{\partial t}$.
\end{itemize}

\paragraph{Steps for calibration}

The calibration is done in a bootstrapping fashion. Let $t_i; i=1, \ldots, n$ be
the increasing sequence of (positive) times where we will perform the
calibration. At each calibration step, we make sure the Feller condition
$2\kappa_{t_i}\theta_{t_i} > \xi_{t_i}^2$ is satisfied
\cite{CIR1985, Feller1951} in order to keep the variance $(U_t)$ strictly
positive, by adding a suitable penalty function to the objective function. 
We choose the simplex algorithm
\cite{NelderMead1965} to do the optimization. The parameters $\kappa_t$,
$\theta_t$, and $\xi_t$ are assumed to be piecewise constant.
\begin{enumerate}
  \item For all calibration time slices $t_i$ and predetermined ranges of
  strikes $K_j$, compute a grid of market vanilla call or put option prices
  using the input market implied volatility $\Sigma(K_j, t_i)$.
  \item Using the PDE pricer, solve for all five parameters $\kappa_{t_1}$,
  $\theta_{t_1}$, $\xi_{t_1}$, $\rho$, and $U_0$ to match the market vanilla
  option prices expiring at $t_1$, generated in the first step.
  \item For the subsequent slices $t_i;\ i > 1$, using the results of the
  previous slices in the PDE pricer, solve for the three parameters
  $\kappa_{t_i}$, $\theta_{t_i}$, and $\xi_{t_i}$ to match the market vanilla
  option prices expiring at $t_i$.
\end{enumerate}
Since our primary goal of calibrating the Heston model
parameters is to recover market quotes around ATMF strikes, we
choose the strike grid to cover a narrower range than in the subsequent
calibration routine of the leverage function.

\subsection{Calibration and Simulation Tests for the Heston model}


The strike grid is chosen uniformly over 1 standard
deviation away in both directions from the ATMF strike value for each
slice. The calibrated Heston parameters obtained by the procedure
outlined above are visualized in Figure \ref{SLVDR_heston_timeseries}.

\begin{figure}[ht!]
    \includegraphics[width=\textwidth]{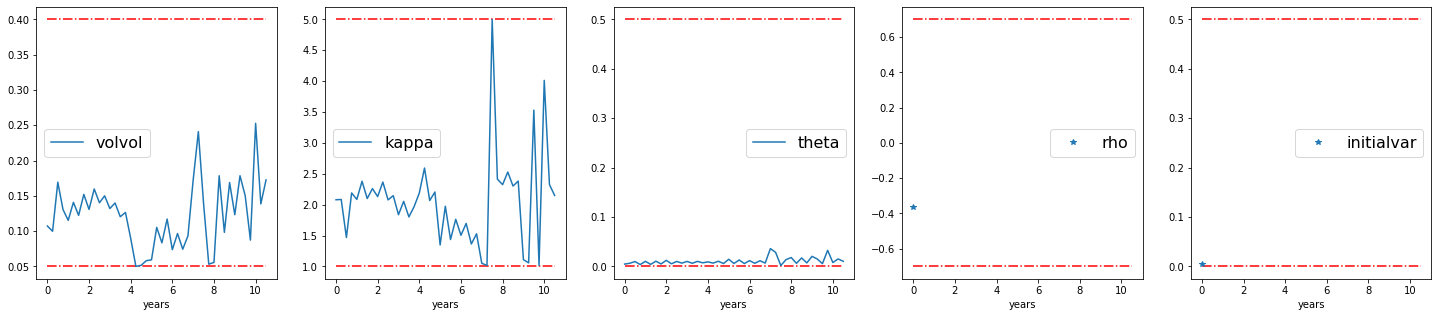}
    \caption{Heston model: Time-Series of calibrated model parameters}
    \label{SLVDR_heston_timeseries}
\end{figure}

\FloatBarrier


The figure shows that the calibrated parameters are mostly within bounds
typically encountered in practice. These parameters will be used in
calibration of the leverage function for the stochastic local volatility
model as outlined in the subsequent section. We test the validity of
the calibrated Heston parameters, which corresponds to the SLV2DR model
(\ref{eqn:SDEs_SLV}) with the leverage function $L(S_t, t)$ set to 1.0,
by pricing vanilla options by simulation. Thus we can examine these call option 
prices close to ATMF. The repriced call option values at 
maturity and the implied volatilities computed from repriced out-of-the-money
call and put options are shown in Figure \ref{SLV2DR_hestonmaturity_repricing}.
It can be seen that the recovered call option prices as well as the implied
volatilities are near market quoted prices and implied
volatilities for strikes within the calibration region.

\begin{figure}[ht!]
    \centerline{\includegraphics[width=\textwidth]{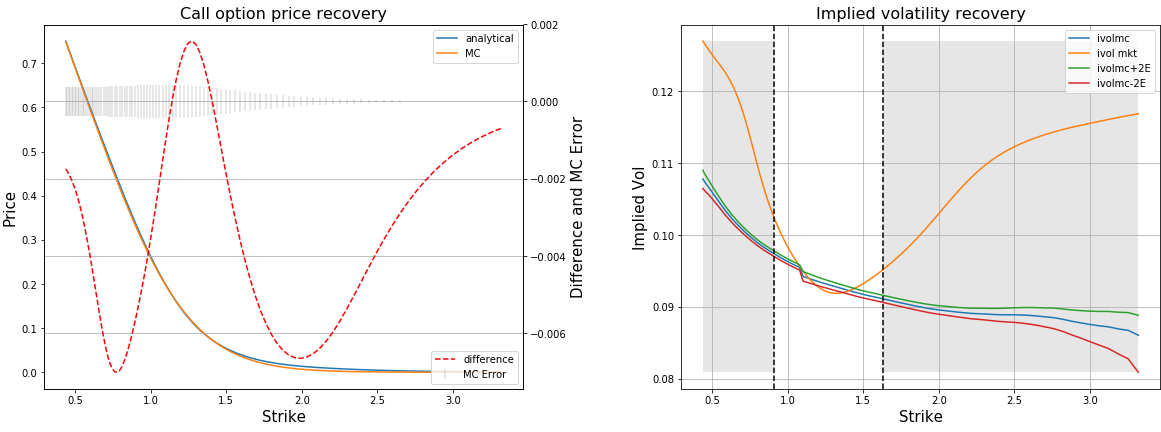}}
    \caption{Heston model: Difference in repriced vs analytical call options
    (left),  implied volatility from repriced out-of-the-money call and put
    options vs market implied
    volatility at maturity(right) (T=9.95). The dashed lines indicate the
    calibration range.}
    \label{SLV2DR_hestonmaturity_repricing}
\end{figure}

\FloatBarrier


We perform an implied volatility recovery test at various maturities by
inverting the simulated Heston model prices to compute the Black-Scholes implied
volatility. We compare this with the market implied volatility in Figure
\ref{SLV2DR_hestonivol_recovery}. The recovered implied volatility is found
to be in good agreement with the market implied volatility for strikes
within the calibration range, as indicated in the plot.

\begin{figure}[ht!]
    \centering
    \centerline{\includegraphics[width=\textwidth]{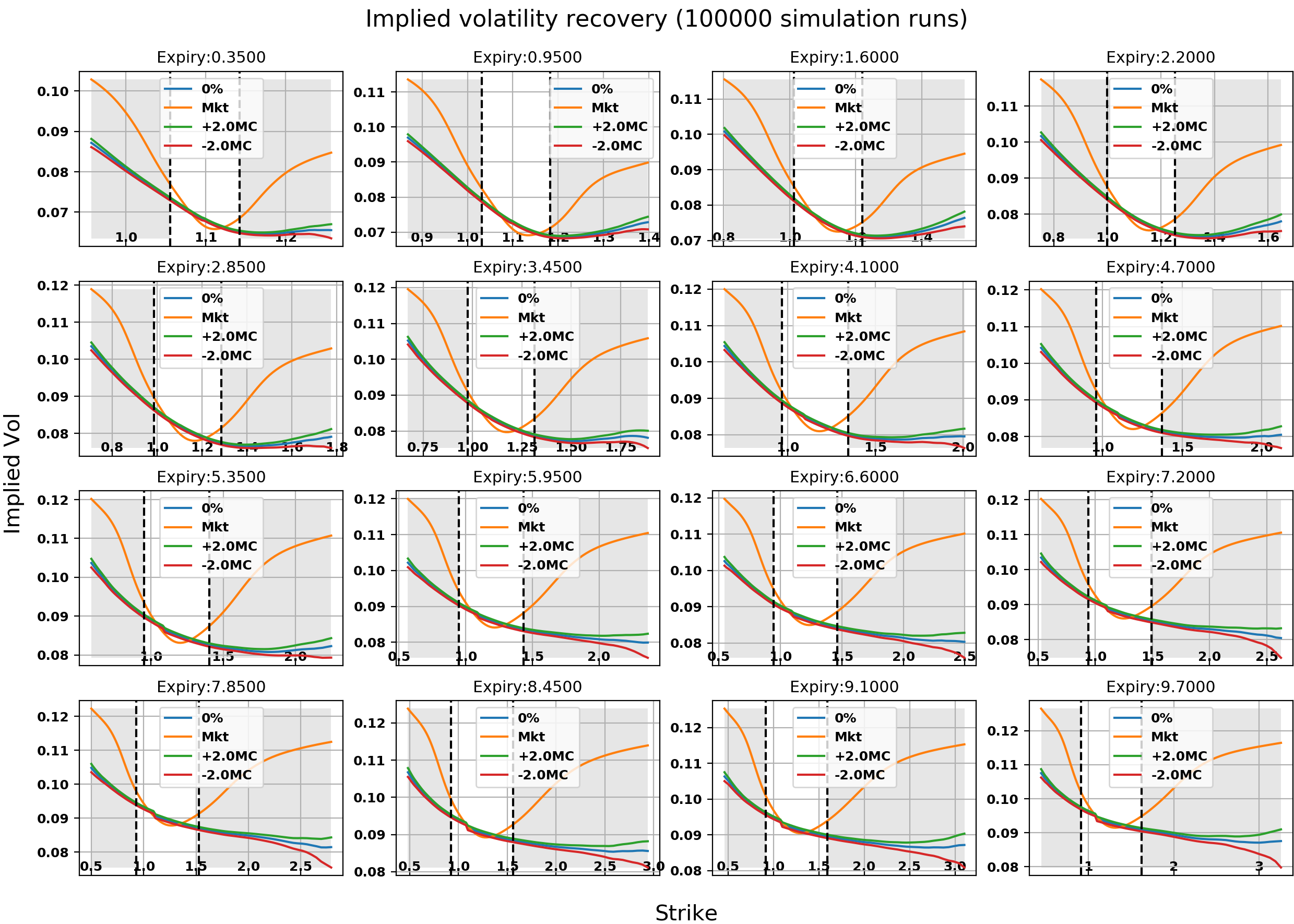}}
    \caption{Heston model: Implied volatility computed from repriced
    out-of-the-money call and put options vs the market implied volatility at
    various maturities. The dashed lines indicate the calibration range at each
    maturity.}
    \label{SLV2DR_hestonivol_recovery}
\end{figure}

\subsection{Calibration of the Leverage Function}

There are various methods proposed in the literature to estimate the conditional
expectation in (\ref{eqn:relation_leverage_locvol}). The standard approach
involves solving a forward Kolmogorov PDE that describes the forward evolution
of the probability density function of the underlier; e.g. the FX rate.
Here we introduce two methods that can be implemented by Monte Carlo simulation.
The Monte Carlo approach we introduce here can be used with problems of higher
dimension; such as a multi-factor stochastic volatility model, or the SLV2SR
model that we will demonstrate in Section \ref{sec:SLV2SR}, where standard PDE
methods cannot be applied due to curse of dimensionality.

\subsubsection{Binning Approach}\label{sec:stoch_locvol_binning}

Suppose we simulate $N$ Monte Carlo paths $(S^i, U^i),\ i = 0, ..., N-1$
up to time $t$. At time $t$, we sort these paths by $S^i$. Let us
denote by $(\hat{S}^i, \hat{U}^i)$ the sorted pairs. We divide these pairs
into $M$ bins, each bin containing $N/M$ pairs. We compute the bin
averages as
\begin{eqnarray*}
   \tilde{S}^j &=& \frac{M}{N} \sum_{k=0}^{\frac{N}{M}-1} \hat{S}^{\frac{N}{M}j + k},\\
   \tilde{U}^j &=& \frac{M}{N} \sum_{k=0}^{\frac{N}{M}-1} \hat{U}^{\frac{N}{M}j + k}.
\end{eqnarray*}
with $j = 0, ..., M-1$. By computing the interpolation function for
$\tilde{S}^j$ against $\tilde{U}^j$, we can estimate the expectation
in (\ref{eqn:relation_leverage_locvol}) for a given $S_t$.

\subsubsection{Regression Approach}

The idea is to linearly regress the variance values $U_t$ against basis
functions $f^n(\cdot)$ of the underlying spot rate values $S_t$.
After simulating $N$ Monte Carlo paths
$(S^i, U^i),\ i = 0, ..., N-1$ up to time $t$, we compute the
regression coefficients $a_n$ by solving the least squares problem
\begin{equation}
   \hat{U}_t = \sum_n a_n f^n(S_t).
   \label{eqn:regression_S_V}
\end{equation}
Standard monomials or orthogonal polynomials with appropriate limits can be used
as basis functions.
For example, if we use a constant term and the first two orders of monomials,
we need to solve
\begin{equation}
   \hat{U}_t = a_1 + a_2 S_t + a_3 S_t^2.
   \label{eqn:regression_S_V_example}
\end{equation}
After computing the regression coefficients $a_n$, we can use this
regression equation to evaluate the expected value of $U_t$ for a given
$S_t$, which gives us the conditional expectation in
(\ref{eqn:relation_leverage_locvol}).

\paragraph{Inputs for calibration}

Our calibration routine expects the following quantities as input for
leverage function calibration:
\begin{itemize}
\item Spot FX rate $S_0$
\item Market implied volatility $\Sigma(K, t)$ for FX rate
\item Market yield curves $P^d(0, t)$ and $P^f(0, t)$. Since the rates are
deterministic, we have $r^i_t = f^i(0, t), i = d, f$ where the instantaneous
forward rate can be computed from the market yield curves, $f^i(0, t) \equiv -
\frac{\partial \log P^i(0, t) }{\partial t}$.
\item Heston model parameters $\kappa_t$, $\theta_t$, $\xi_t$, $\rho$, and $U_0$
calibrated to market vanilla option prices as in Section
\ref{sec:heston_calib}. 
The strike grid is chosen to be near ATMF strikes. In
particular it should be smaller than (e.g. one third of) the local 
volatility strike range.
\end{itemize}

\paragraph{Steps for calibration}

The calibration is done in a bootstrapping fashion. Let $t_i; i=1, \ldots, n$ be
the increasing sequence of (positive) times where we will perform the
calibration. After computing the leverage function values for a time slice
$t_j$, the values are used during the subsequent simulations to estimate
the leverage function values at the next time slice.
\begin{enumerate}
  \item Using the market implied volatility $\Sigma(K, t)$, generate a vanilla
  call option price surface $C(K, t)$ interpolator or a total implied variance
  surface $w(y, t)$ interpolator. The interpolator must be able to evaluate the
  partial derivatives appearing in the local volatility expressions.
  \item Compute the deterministic local volatility $\sigma_{\text{LV}}(K, t)$
  by (\ref{eqn:dupire_C_deterministic_ir}) or
  (\ref{eqn:dupire_tiv_deterministic_ir}) for all calibration time slices $t_i$
  and predetermined ranges of strikes.
  \item Simulate the SDE system (\ref{eqn:SDEs_SLV}) up to time $t_j$. Compute
  the Monte Carlo estimate for the conditional expectation appearing in
  (\ref{eqn:relation_leverage_locvol}) for the same range of strikes from the
  previous step by using one of the approaches described in the previous
  sections. Use this equation to obtain the leverage function values $L(K, t)$.
  These leverage function values will be used during subsequent simulation
  steps from time $t_j$ to time $t_{j + 1}$.
\end{enumerate}
The strike grid can be chosen to be uniform across
all calibration time slices.
In this approach, however, the strike grid needs to be sufficiently large to
cover attainable values of the FX rate at long expiries. This in turn would
result in unreliable local volatility values at short expiries and strikes in
the far wings. To overcome this problem, we suggest using a more adequate
strike grid at each calibration time slice, e.g. one that spans a predetermined
number of standard deviations away from the ATMF strike value. Another plausible
approach is choosing the strike grid to cover the range of strikes that the
implied volatility surface is calibrated to, if this information is available.
For the Heston parameters calibration, we suggest using a similar grid with a
smaller range, e.g. one third of the local volatility grid range.

\subsection{Calibration and Simulation Tests for the Leverage Function}


In order to test the calibration convergence, the differences
between the calibrated leverage functions with 100,000 and 50,000 MC
paths are visualized in Figure \ref{SLVDR_convergence} across the strike
grid at various time slices. As the figure shows, the differences are minor,
indicating the achievement of convergence.

\begin{figure}[ht!]
    \centering \includegraphics[width=\textwidth]{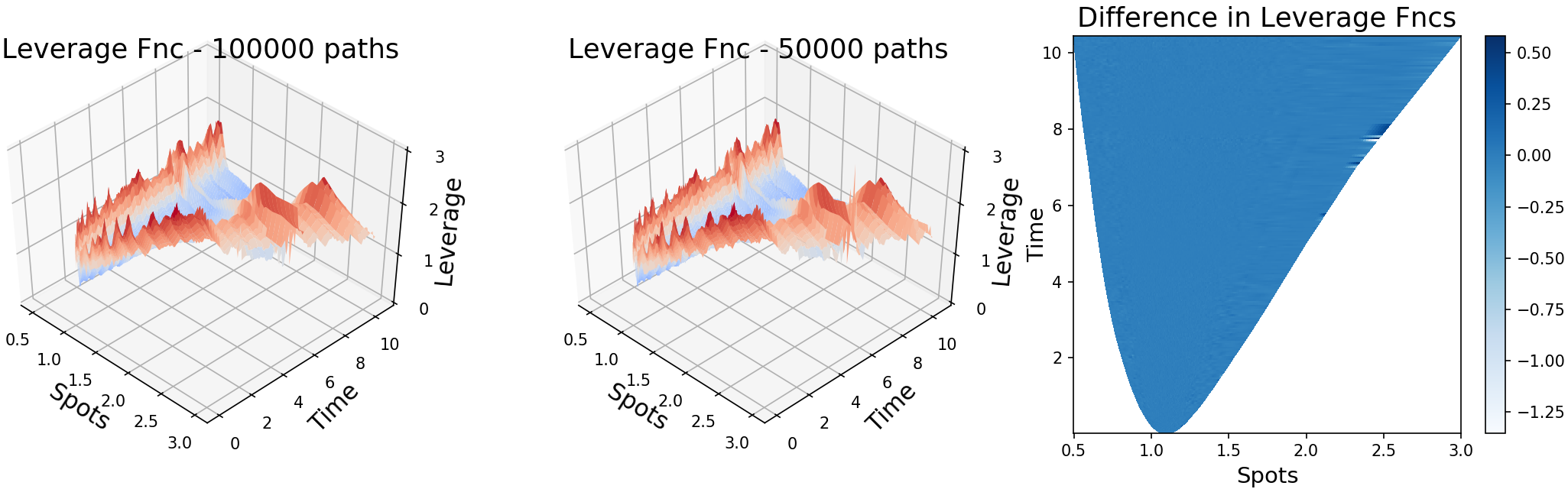}
    \caption{SLV2DR: Difference in leverage function with respect to the
    calibration Monte Carlo paths}
    \label{SLVDR_convergence}
\end{figure}

\FloatBarrier


The leverage function calibrated with 100,000 MC paths was used
for all subsequent tests including repricing and implied volatility recovery.
The call option was repriced for strikes at maturity with the above leverage
function, with various number of simulation MC paths as shown in Figure
\ref{SLV2DR_maturity_repricing}. As the MC error
decreases with the number of simulation paths, between 1,000 and 100,000
paths the maximum absolute difference between the analytical and Monte Carlo priced call
option values decreases from $0.0128$ to $0.0015$.

\begin{figure}[ht!]
    \centering \includegraphics[width=\textwidth]{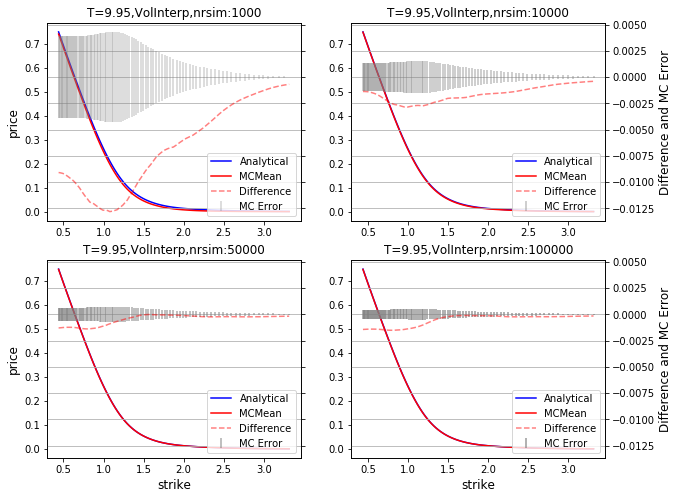}
    \caption{SLV2DR: Repricing the call options with different simulation
    Monte Carlo paths using leverage function calibrated with 100,000
    Monte Carlo paths}
    \label{SLV2DR_maturity_repricing}
\end{figure}

\FloatBarrier


We perform a simulation to reprice vanilla call options across the strike grid
at various maturities to obtain the so called call surface. The
difference between the analytical call price surface implied by $\Sigma(K, t)$
and repriced surface using 100,000 MC paths at each strike and maturity is
visualized in Figure \ref{fig:SLV_2DRcallsurface}. The differences appear to be
small compared to option prices.

\begin{figure}[ht!]
\begin{subfigure}{.48\textwidth}
  \centering
  \includegraphics[width=.9\linewidth]{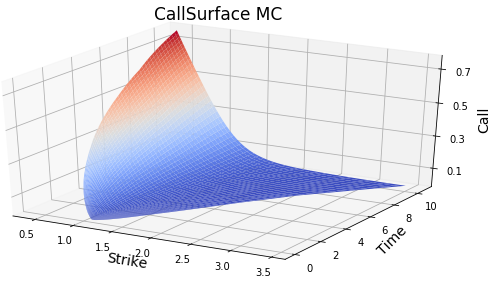}  
  \caption{SLV2DR: Repriced call surface at all strikes and maturities in the grid}
\end{subfigure}
\begin{subfigure}{.48\textwidth}
  \centering
  \includegraphics[width=.95\linewidth]{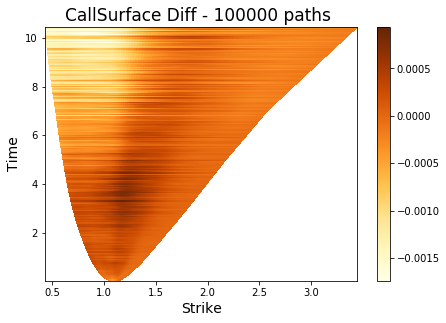}  
  \caption{SLV2DR:Difference between analytical and MC repriced call
  surfaces}
\end{subfigure}
\caption{SLV2DR: Call options repriced at all strikes and maturities and the
difference from BS analytical price}
\label{fig:SLV_2DRcallsurface}
\end{figure}

\FloatBarrier


The implied volatility for the repriced out-of-the-money call and put options at
maturity was obtained by inverting option values with Black-Scholes formula.
The recovered implied volatility is compared to market implied volatility
as shown in Figure \ref{SLV2DR_maturity_ivol}. This is also
performed for option prices obtained by bumping them by $\pm 2$ MC errors.
The market implied volatility and the recovered implied volatility are found to
be in good agreement with each other.

\begin{figure}[ht!]
    \centerline{\includegraphics[width=\textwidth]{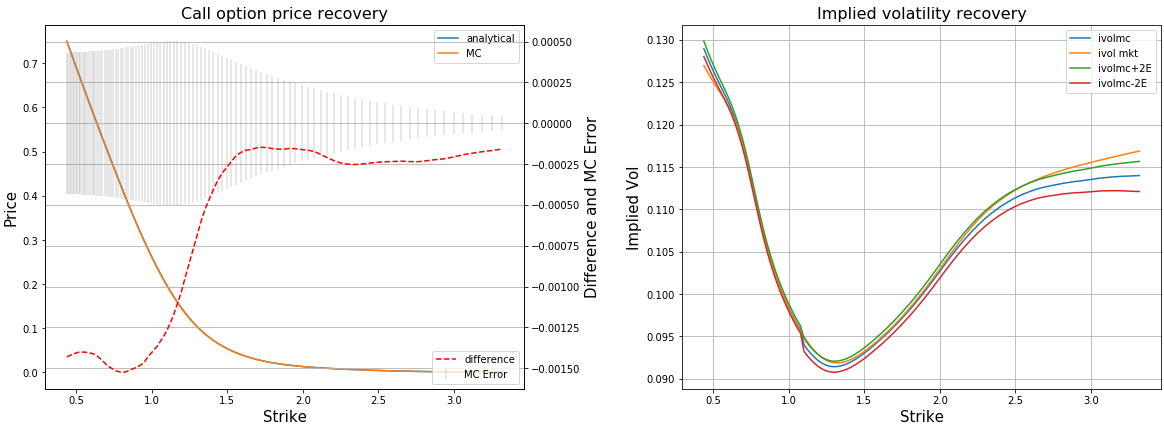}}
    \caption{SLV2DR: Difference in repriced vs analytical call options (left),
    implied volatility from repriced out-of-the-money call and put option vs
    market implied volatility at maturity(right) (T=9.95)}
    \label{SLV2DR_maturity_ivol}
\end{figure}

\FloatBarrier


Finally, the same procedure of recovering implied volatility with
the repriced out-of-the-money call and put option values, is repeated at
multiple maturities and strikes
for different time slices, as shown in Figure \ref{SLV2DR_ivol_recovery}, which
is also found to be in good agreement with the given market implied volatility.
We note that our implied volatility recovery results are consistent with those
demonstrated in \cite{RenMadanQian2007} where the authors calibrate the
stochastic local volatility model by solving a two-dimensional forward
Kolmogorov PDE.

\begin{figure}[ht!]
    \centerline{\includegraphics[width=\textwidth]{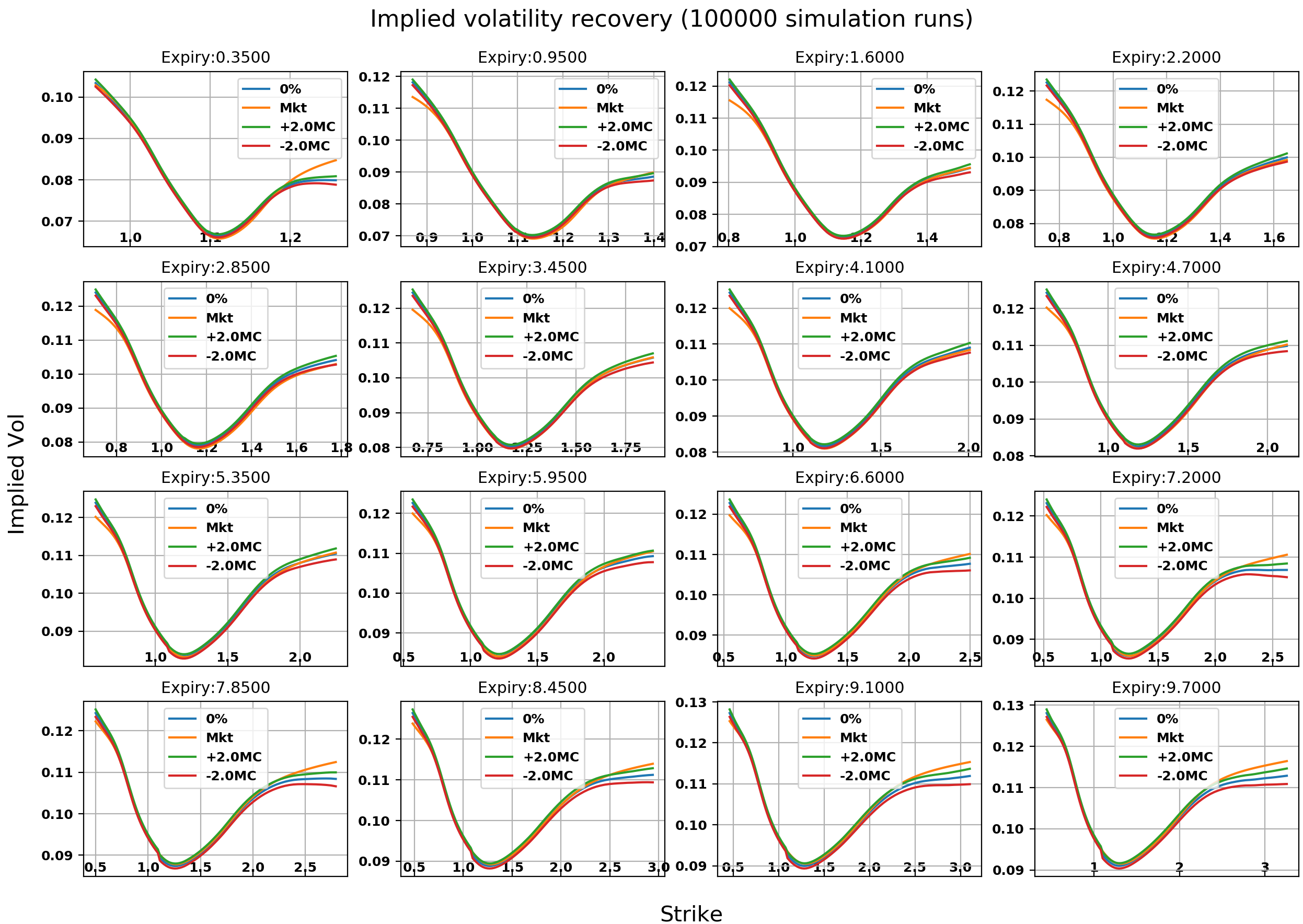}}
    \caption{SLV2DR: Implied volatility computed from repriced out-of-the-money
    call and put options vs the market implied volatility at various maturities}
    \label{SLV2DR_ivol_recovery}
\end{figure}

We have observed in Figures \ref{SLV2DR_hestonmaturity_repricing}
and \ref{SLV2DR_hestonivol_recovery} that the pure Heston model recovered market
quotes near ATMF strikes. Comparison of these to Figures
\ref{SLV2DR_maturity_repricing} and \ref{SLV2DR_ivol_recovery} shows us the
improvement of the full stochastic local volatility model over the Heston model
at all strikes. In particular, we see that the pricing error was reduced by a
factor of 4, and the market implied volatilities are recovered along a greater
range.

\section{Stochastic Local Volatility Model with Stochastic Interest Rates
(SLV2SR)}
\label{sec:SLV2SR}

\subsection{Setup / Model Definition}

This model can be seen an extension to both local volatility model with
stochastic interest rates (LV2SR) and stochastic local volatility model with
deterministic interest rates (SLV2DR). Considering an FX rate process with
stochastic local volatility and stochastic interest rates, the SLV2SR SDE system
can be written in the domestic risk neutral measure $\mathbb{Q}^{\text{DRN}}$ as
\begin{align}
\begin{split}
   dS_t =& \left[r^d_t - r^f_t\right] S_t dt + L(S_t, t)
   \sqrt{U_t} S_t dW_t^{S\text{(DRN)}},\\
   dU_t =& \kappa_t (\theta_t - U_t) dt + \xi_t \sqrt{U_t}
   dW_t^{U\text{(DRN)}},\\
   r^d_t =& x^d_t + \phi^d_t, \\
   dx^d_t =& -a^d_t x^d_t dt + \sigma^d_t dW_t^{d\text{(DRN)}},\\
   r^f_t =& x^f_t + \phi^f_t, \\
   dx^f_t =& \left[-a^f_t x^f_t - \rho_{Sf} \sigma^f_t L(S_t, t)
   \sqrt{U_t}\right] dt + \sigma^f_t dW_t^{f\text{(DRN)}},
   \label{eqn:SDE_stochlv_stochir}
\end{split}
\end{align}
with coefficients of correlation $\rho_{SU}$ (FX rate/FX variance),
$\rho_{Sd}$ (FX rate/domestic interest rate), $\rho_{Sf}$ (FX rate/foreign
interest rate), $\rho_{Ud}$ (FX variance/domestic interest rate), $\rho_{Uf}$
(FX variance/foreign interest rate), and $\rho_{df}$ (domestic interest
rate/foreign interest rate) between the respective Brownian motions.

One can relate the standard local volatility $\sigma_{\text{LV}}(S_t, t)$
to the leverage function $L(S_t, t)$ by \cite{Ogetbil2020}
\begin{equation}
   \sigma_{\text{LV}}(x, t)^2 = L(x, t)^2 \mathbf{E}^{\mathbb{Q}^{\text{T}}}
   \left[U_t \mid S_t=x \right].
   \label{eqn:relation_leverage_locvol_xyz}
\end{equation}
Since our calibration algorithm presented in the next section demands Monte
Carlo estimation of the above expectation in the domestic $T$-forward measure
$\mathbb{Q}^{\text{T}}$, we need to formulate the SDE system in this measure.
The $T$-forward measure SDEs for the exchange rate, domestic short rate and
domestic short rate are derived the same way as in Section \ref{sec:Tforward}.
The remaining computation is for the stochastic variance SDE, which we present
here.

Let $\rho_{Ud}$ is coefficient of correlation between the Brownian motions
$(W^{U\text{(DRN)}}_t)$ and $(W^{d\text{(DRN)}}_t)$, that is
$
d\left<W^{U\text{(DRN)}}, W^{d\text{(DRN)}} \right>_t = \rho_{Ud} dt.
$
The domestic risk neutral measure $\mathbb{Q}^{\text{DRN}}$ to domestic
$T$-forward measure $\mathbb{Q}^{\text{T}}$ transformation is characterized by
(\ref{eqn:dW_dT_dDRN}), Following Lemma \ref{lemma:measure_trf}, the variance
process (\ref{eqn:SDE_stochlv_stochir}) evolves in domestic $T$-forward measure
$\mathbb{Q}^{\text{T}}$ as
\begin{equation}
dU_t = \left[\kappa_t (\theta_t - U_t) - \rho_{Ud} b^d(t, T) \sigma^d_t \xi_t
\sqrt{U_t} \right] dt + \xi_t \sqrt{U_t} dW^{U\text{(T)}}_t.
\end{equation}

Collecting everything,
\begin{align}
\begin{split}
   dS_t =& \left[r^d_t - r^f_t - \rho_{Sd} b^d(t, T) \sigma^d_t L(S_t, t)
   \sqrt{U_t}\right] S_t dt + L(S_t, t) \sqrt{U_t} S_t dW_t^{S\text{(T)}},\\
   dU_t =& \left[\kappa_t (\theta_t - U_t) - \rho_{Ud} b^d(t, T) \sigma^d_t
   \xi_t \sqrt{U_t} \right] dt + \xi_t \sqrt{U_t} dW^{U\text{(T)}}_t,\\
   r^d_t =& x^d_t + \phi^d_t, \\
   dx^d_t =& \left[-a^d_t x^d_t - b^d(t, T) (\sigma^d_t)^2 \right] dt +
   \sigma^d_t dW_t^{d\text{(T)}},\\
   r^f_t =& x^f_t + \phi^f_t, \\
   dx^f_t =& \left[ -a^f_t x^f_t - \rho_{Sf} \sigma^f_t L(S_t, t)
   \sqrt{U_t} - \rho_{df} b^d(t, T) \sigma^d_t \sigma^f_t \right] dt +
   \sigma^f_t dW_t^{f\text{(T)}}
   \label{eqn:SDE_stochlv_stochir_Tfwd}
\end{split}
\end{align}
describe the evolutions of the exchange rate, exchange rate variance, domestic
short rate, and foreign short rate processes under the domestic $T$-forward
measure $\mathbb{Q}^{\text{T}}$.

\subsection{Calibration of the Leverage Function}

The standard forward Kolmogorov PDE approach to solve the conditional
expectation in (\ref{eqn:relation_leverage_locvol_xyz}) suffers from the curse
of dimensionality, as this is now a 4D problem. Similarly, the binning approach
utilized in Section \ref{sec:stoch_locvol_binning} for the model with
deterministic interest rates can not be applied directly, at least without any
simplifying assumptions, as the sorting of the underliers becomes nontrivial.

As in the case of stochastic local volatility model with deterministic interest
rates, the calibration is done in a bootstrapping fashion, after computing the
leverage function values for a time slice $t$, the values are used during
the subsequent simulation to estimate the leverage function values at the next
time slice. The Heston model parameters are assumed to be calibrated to match
an appropriate subset of market data.

The idea is to linearly regress the variance values $U_t$ against basis
functions $f^n(\cdot)$ of the underlying spot rate values $S_t$,
and the two interest rate values $r_t^d$ and $r_t^f$.
After simulating $N$ Monte Carlo paths
$(S^i, r^{d,i}, r^{f,i}, U^i),\ i = 0, ..., N-1$ up to time $t$, we
compute the regression coefficients $a_n$ by solving the least squares
problem
\begin{equation}
   \hat{U}_t = \sum_n a_n f^n(S_t, r_t^d, r_t^f).
   \label{eqn:regression_S_rd_rf_U}
\end{equation}
Standard monomials or orthogonal polynomials can be used as basis functions.
For example, if we use a constant term and the first two orders of monomials
for all underliers, we need to solve
\begin{equation}
   \hat{U}_t = a_1 + a_2 S_t + a_3 S_t^2 + a_4 x_t^d + a_5 {x_t^d}^2 + a_6 x_t^f
   + a_7 {x_t^f}^2.
   \label{eqn:regression_S_rd_rf_U_example}
\end{equation}

Note that in this example we used the $x_t^d$ and $x_t^f$ as the
basis functions instead of the short rates $r_t^d$ and $r_t^f$;
the deterministic parts $\phi_t^d$ and $\phi_t^f$ of the latter
can be absorbed into other coefficients.
After computing the regression coefficients $a_n$, we can use this
regression equation to evaluate the expected value of $U_t$ for given
$S_t$, $r_t^d$, and $r_t^f$, which gives us the conditional
expectation in (\ref{eqn:relation_leverage_locvol_xyz}).

\paragraph{Inputs for calibration}

Our calibration routine expects the following quantities as input for leverage
function calibration:
\begin{itemize}
\item Spot FX rate $S_0$
\item For both domestic and foreign rates, G1++ model parameters mean
reversion, volatility and shift function calibrated to market data\footnote{See
\cite{GurrieriNakabayashiWong2009} for example calibration methods for both
constant and time dependent cases.}
\item Coefficients of correlation between all underlying assets: the FX rate,
its variance, the domestic and foreign short rates
\item Local volatility (with stochastic rates) surface data, as calibrated in
Section \ref{sec:calib_lv_sr}
\item Heston model (with deterministic rates) parameters, as calibrated in
Section \ref{sec:heston_calib}
\end{itemize}

\paragraph{Steps for calibration}

In our framework, we calibrate the leverage function surface time slice by
time slice, in a bootstrapping fashion. Let $t_i; i=0, \ldots, n$ be the
increasing sequence of times where we will perform the calibration. The first
time slice is $t_0 = 0$.
\begin{enumerate}
  \item For $t_0$, (\ref{eqn:relation_leverage_locvol_xyz}) simplifies to
  $\sigma_{\text{LV}}(K, 0)^2 = L(K, 0)^2 U_0$, where $\sigma_{\text{LV}}(K,
  t)^2$ is the local volatility with stochastic rates. Use this equation to
  evaluate the leverage function for the $t_0$ slice for a predetermined range
  of strikes.
  \item For each of the subsequent positive time slices $t_j, j \geq 1$,
  simulate the SDE system (\ref{eqn:SDE_stochlv_stochir_Tfwd}) up to time $t_j$.
  Compute the Monte Carlo estimate for the expectation appearing in
  (\ref{eqn:relation_leverage_locvol_xyz}) for a predetermined range of strikes.
  Use this equation to obtain the leverage function values. These leverage
  function values will be used during subsequent simulation steps from time
  $t_j$ to time $t_{j + 1}$.
\end{enumerate}
The strike grid for the leverage function $L(K, t)$ is chosen to be the same as
the strike grid for the local volatility $\sigma_{\text{LV}}(K, t)$ to avoid any
inaccuracy introduced by interpolation.

\subsection{Calibration and Simulation Tests}

The G1++ model parameters we use as input are summarized in Table
\ref{table:g1pp_parameters}. 
The coefficients of correlation are given by $\rho_{Sd} = 0.166, \rho_{Sf} =
0.551, \rho_{df} = 0.161$.

In order to calibrate the SLV2SR model, the Heston model parameters computed
during the SLV2SR model calibration and the local volatility function
calibrated for the LV2SR model with 100,000 MC paths were used. The
procedure outlined above was followed in order to iteratively compute the
leverage function for each time slice. In order to evaluate the expectation
appearing in (\ref{eqn:relation_leverage_locvol_xyz}), 100,000 MC paths followed
by regression with monomials of order 2 for each of the underlying regressors
$(S_t, x_t^d, x_t^f)$ along with the constant coefficient $a_1$ as in
(\ref{eqn:regression_S_rd_rf_U_example}) were used. The results for the
calibration and repricing tests are presented below.

The effect of number of repricing simulation paths for the given
local volatility surface is shown in Figure \ref{SLV2SR_maturity_repricing}
against the call option price and the corresponding MC error. As the MC error
decreases with the number of simulation paths, between 1,000 and 100,000
paths the maximum absolute difference between the analytical and Monte Carlo priced call
option values decreases from $0.00437$ to $0.00146$.

\begin{figure}[ht!]
    \centering \includegraphics[width=\textwidth]{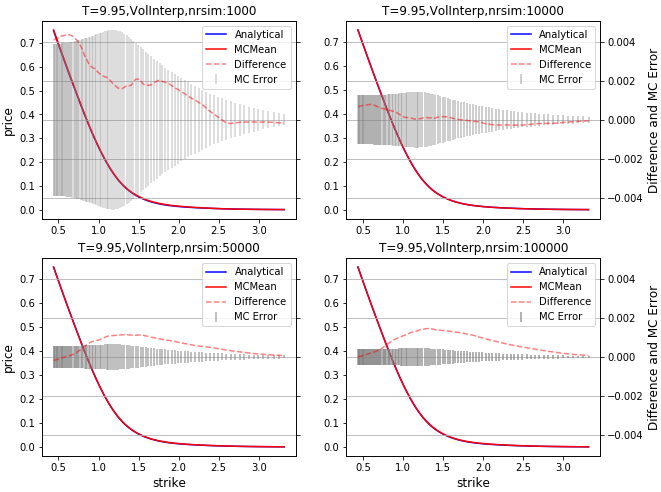}
    \vspace{-25pt}
    \caption{SLV2SR: Repricing call options with various simulation
    Monte Carlo paths using leverage function calibrated for SLV2DR with
    100,000 MC paths and multi-regression approach with 100,000 MC paths.}
    \label{SLV2SR_maturity_repricing}
\end{figure}

\FloatBarrier


Furthermore, the local volatility surface was used to reprice
call options at multiple maturities and various strikes in the strike-grid to
generate the so-called call price surface.
100 maturies uniformly spaced between T=0 and 9.95 years, and 100 strikes per
maturity were used to generate the call price surface shown in Figure
\ref{fig:slv2srcallsurface}a.
The difference between the Monte Carlo repriced call option values and
analytical Black-Scholes call option prices, implied by $\Sigma(K, t)$, assuming constant
interest rate and volatilty is shown in Figure \ref{fig:slv2srcallsurface}b,
which is found to be less than 1\% of the call option price.

\begin{figure}[ht!]
\begin{subfigure}{.49\textwidth}
  \centering
  \includegraphics[width=\linewidth]{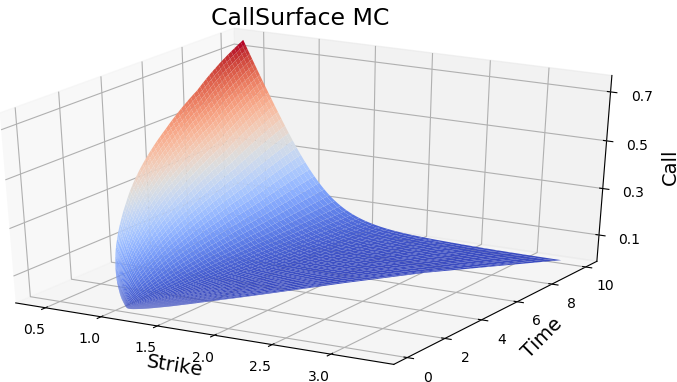}  
  \caption{SLV2SR: Repriced call surface at all strikes and maturities in the grid}
\end{subfigure}
\begin{subfigure}{.49\textwidth}
  \centering
  \includegraphics[width=\linewidth]{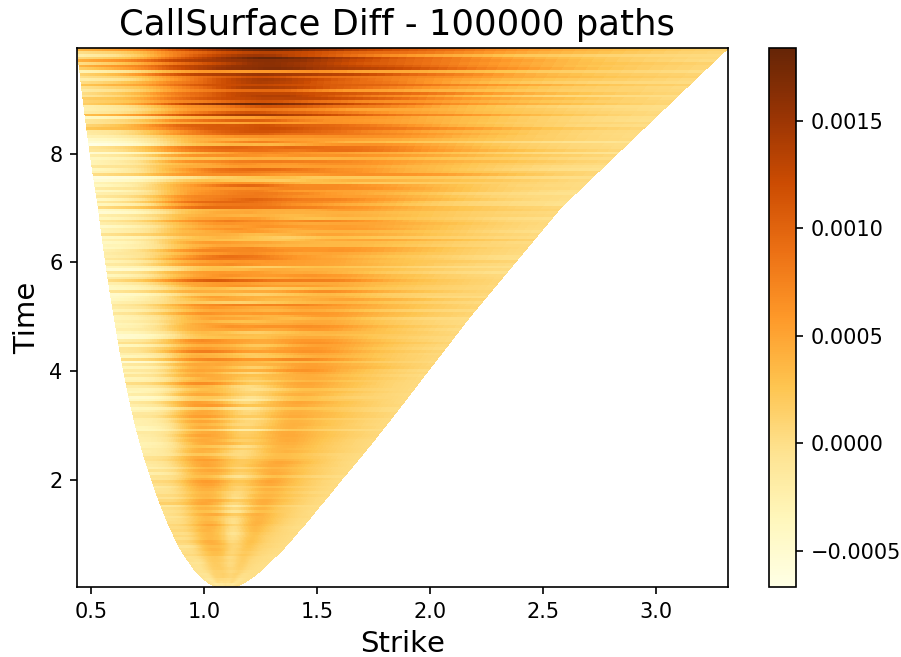}  
  \caption{SLV2SR: Difference between analytical and MC repriced call
  surfaces}
\end{subfigure}
\caption{SLV2SR: Call options repriced at all strikes and maturities and the
difference from BS analytical price}
\label{fig:slv2srcallsurface}
\end{figure}

\FloatBarrier

Next, the market implied volatility and the implied volatility
recovered from the Monte Carlo repriced out-of-the-money call and put options,
by inverting the Black-Scholes formula, are compared in Figure
\ref{SLV2SR_maturity_repricing2} at maturity T=9.95. In addition the implied
volatility recovered from MC prices $\pm 4$ MC
errors are presented. It can be seen that recovered implied volatility
is in good agreement with market implied volatility.

\begin{figure}[ht!]
    \centerline{\includegraphics[width=\textwidth]{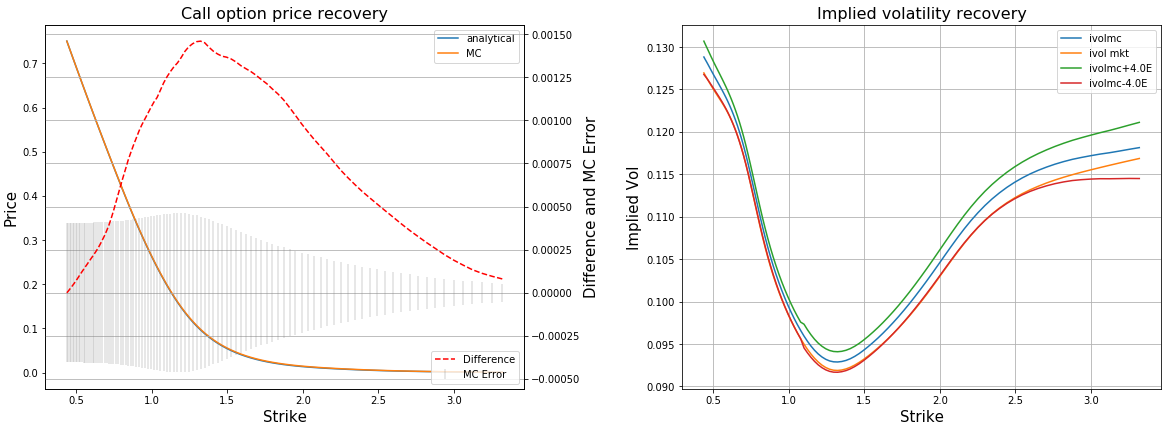}}
    \caption{SLV2SR: Difference between MC repriced and analytical call options (left),
    implied volatility from repriced out-of-the-money call and put options vs
    market implied volatility at maturity(right) (T=9.95)}
    \label{SLV2SR_maturity_repricing2}
\end{figure}


\FloatBarrier


Finally, this procedure is repeated for all the
maturities (time slices) in the repriced call surface, where the market and
recovered implied volatility along with implied volatilities corresponding to
$\pm 4$ MC pricing errors for a few of the slices are shown in
Figure \ref{SLV2SR_ivol_recovery}.

\begin{figure}[ht!]
    \centerline{\includegraphics[width=\textwidth]{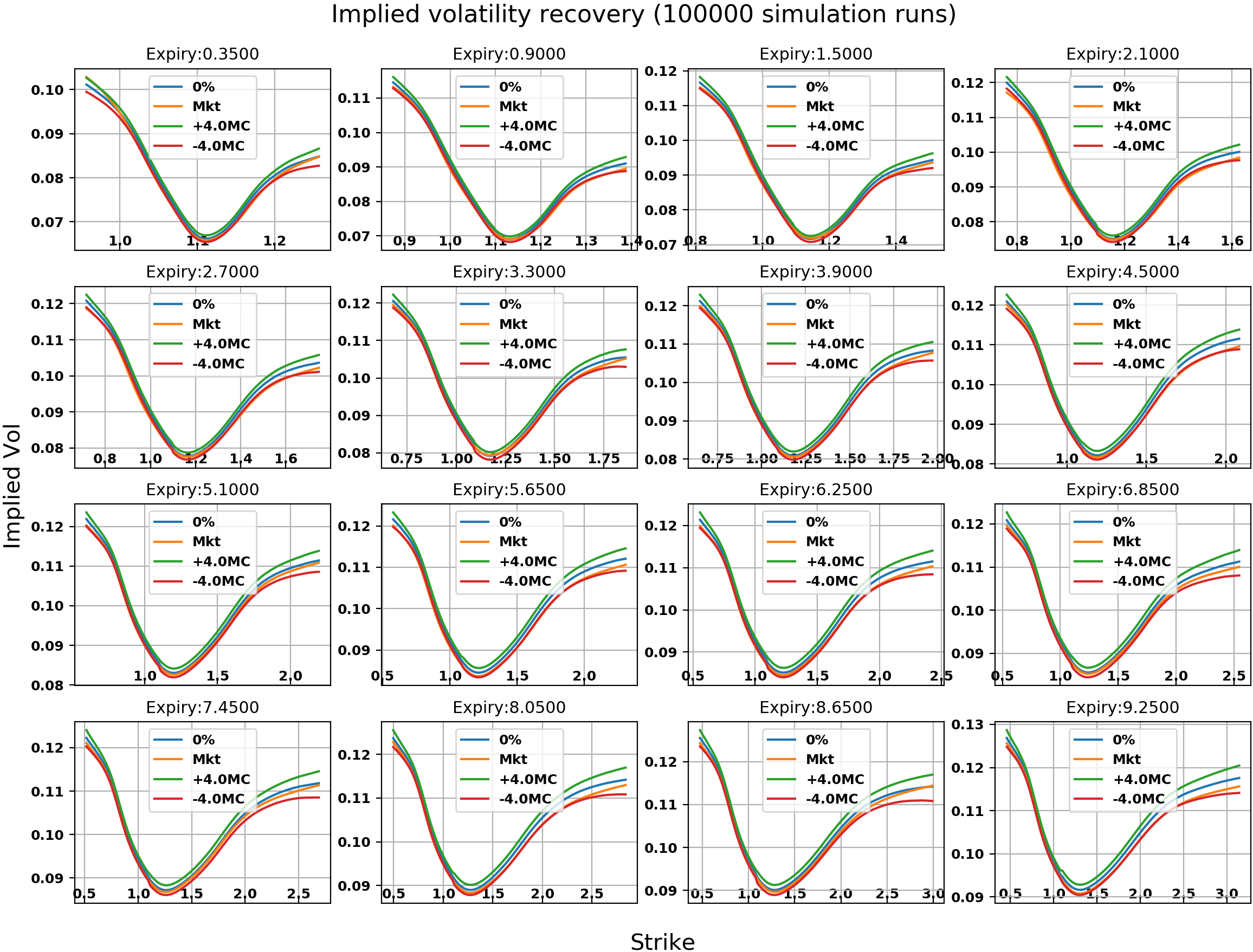}}
    \caption{SLV2SR: Implied volatility computed from repriced out-of-the-money
    call and put option vs the market implied volatility at various maturities}
    \label{SLV2SR_ivol_recovery}
\end{figure}

\FloatBarrier

\section{Further Studies and Conclusions}\label{sec:conclusion}

We studied the convergence and the vanilla option repricing accuracy of the
LV2SR, SLV2DR and SLV2SR models calibrated with the proposed algorithms. While
all three models perform decently with recovering market implied volatilities,
we found that as the models get more complex, e.g. when they have higher number
of parameters, one needs to increase the number of simulation paths to maintain
the accuracy, which in turn results in increased calibration time.

We are now in a position to assess the pricing inaccuracies of the three
main models (LV2SR, SLV2DR and SLV2SR) we considered in this paper. 
To gain a more comprehensive perspective, we also simulate the standard local
volatility model with two interest rates LV2DR (\ref{eqn:sde_stdlv_domrn}) where
the local volatility is computed from market implied volatility by
(\ref{eqn:dupire_tiv_deterministic_ir}). For the LV2DR model, Figure
\ref{LV2DR_maturity_repricing} demonstrates the convergence of the Monte Carlo
price with respect to the number of simulated paths, and Figure
\ref{fig:lv2drcallsurface} shows the Monte Carlo simulated call price surface
and its difference from the analytical call price surface implied by market
data.
\begin{figure}[ht!]
    \centering \includegraphics[width=\textwidth]{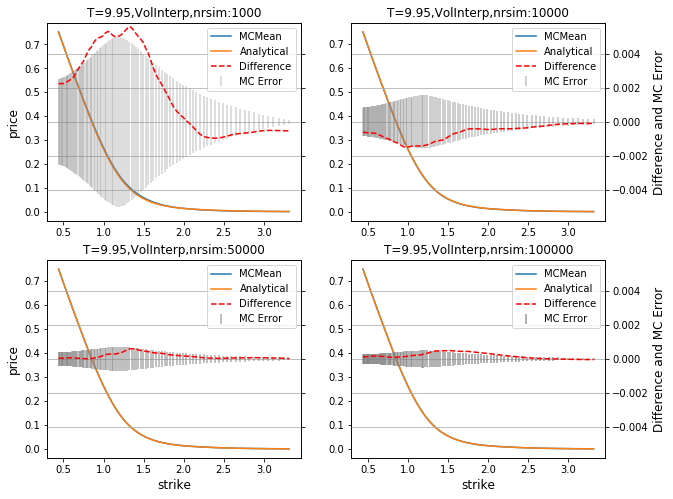}
    \caption{LV2DR: Repricing call options at 100 uniformly spaced strikes at
    maturity T=9.95 years, each with 100,000 (lower-right), 50,000 (lower-left), 10,000
    (upper-right) and 1,000 (upper-left) MC simulation paths,
    using the local volatility surface calibrated with 100,000 MC simulation
    paths. The MC errors and the difference between the analytical and MC
    computed prices decreases with the number of simulation paths.}
    \label{LV2DR_maturity_repricing}
\end{figure}
\begin{figure}[ht!]
\begin{subfigure}{.49\textwidth}
  \centering
  \includegraphics[width=\linewidth]{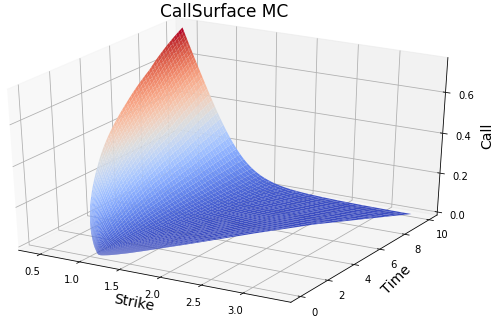}  
  \caption{Repriced call option surface at all strikes and maturities in the grid}
\end{subfigure}
\begin{subfigure}{.49\textwidth}
  \centering
  \includegraphics[width=\linewidth]{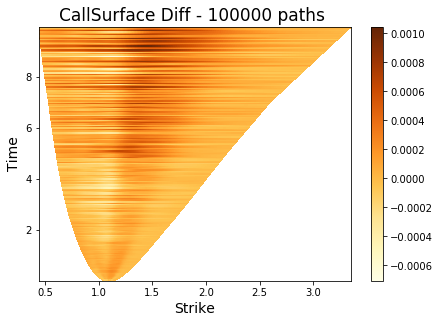}  
  \caption{Difference in analytical call surface and repriced call surface}
\end{subfigure}
\caption{LV2DR: (a) Call options repriced at 100 uniformly spaced maturities
between T=0 to 9.95 years and 100 strikes per maturity and (b) the difference
between Monte Carlo and Black-Scholes analytical price.}
\label{fig:lv2drcallsurface}
\end{figure}

With fixed
number of calibration and simulation paths, we observe that the SLV2DR model
(see Figures \ref{SLV2DR_maturity_repricing}, \ref{fig:SLV_2DRcallsurface},
\ref{SLV2DR_maturity_ivol}) reprices market quotes slightly more
accurately than the SLV2SR model (see Figures
\ref{SLV2SR_maturity_repricing}, \ref{fig:slv2srcallsurface},
\ref{SLV2SR_maturity_repricing2}). In the meantime, we see that our simplest
model, LV2SR (see Figures
\ref{LV2SR_maturity_repricing}, \ref{fig:lv2srcallsurface},
\ref{LV2SR_maturity_repricing2}) gives the smallest repricing errors among the
three models. In the meantime, we observe that the accuracy of LV2SR
calibration is comparable to, and arguably slightly better than, that of LV2DR.

Recovering market quotes is clearly a property desired for any pricing model. In
our case, the market quotes are vanilla option prices or an implied volatility
surface. Yet the risk neutral price of a European vanilla option only depends
on the terminal distribution of the underlier, which can be extracted
from the implied volatility surface numerically. Thus, the three models we are
considering are not strictly necessary to price vanilla options. What about
instruments for which the payoff depends on the joint distribution of
intermediate values of the underliers?

To study the pricing of path dependent options,
we consider an up-and-out barrier call option with 5-year maturity struck at
ATMF strike. Without the barrier feature, the instrument becomes a plain
European vanilla call option, which has an analytical solution under the
Black-Scholes model. The models under consideration price the option with
valuation date 2020-04-30 and 100,000 simulation paths as in Table
\ref{table:vanilla_repricing}.
\begin{table}[ht!]
\footnotesize
\begin{center}
\caption{\footnotesize Vanilla option pricing with various
models. Prices and errors are given in basis
points.}\label{table:vanilla_repricing}
\begin{tabular}{lcc}
 \hline
 \multicolumn{1}{c}{Model} & \multicolumn{1}{c}{Price ($\times 10^{-4}$)} &
 \multicolumn{1}{c}{Error ($\times 10^{-4}$)} \\
 \hline
 Analytical (BS) & 843.79 &  \\
 LV2DR & 845.53 & 2.93\\
 LV2SR & 846.33 & 2.86 \\
 SLV2DR & 843.36 & 2.93 \\
 SLV2SR & 846.51 & 2.89 \\
 \hline
\end{tabular}
\end{center}
\end{table}
We note that the analytical price is within the Monte Carlo error for all the
four models, in consistency with the findings of the previous sections.

Now we turn the up-and-out barrier on and set the barrier position at 1.25 times
the ATMF strike.
The barrier is active throughout the lifetime of the option. The standard
Black-Scholes model admits an analytical solution for such barriers when the
interest rates are constant \cite{Wilmott2013}. The models under consideration
price the option with the same valuation date and number of simulation
paths as in Table \ref{table:barrier_repricing}.
\begin{table}[ht!]
\footnotesize
\begin{center}
\caption{Barrier option pricing with various
models. Prices and errors are given in basis
points.}\label{table:barrier_repricing}
\begin{tabular}{lcc}
 \hline
 \multicolumn{1}{c}{Model} & \multicolumn{1}{c}{Price ($\times 10^{-4}$)} &
 \multicolumn{1}{c}{Error ($\times 10^{-4}$)} \\
 \hline
 Analytical (BS) & 285.55 &  \\
 LV2DR & 287.05 & 1.08 \\
 LV2SR & 284.87 & 1.07 \\
 SLV2DR & 299.00 & 1.14 \\
 SLV2SR & 300.84 & 1.12 \\
 \hline
\end{tabular}
\end{center}
\end{table}
We find that the LV2DR and LV2SR model prices are within 2 Monte Carlo errors of
the analytical price. However the SLV2DR and SLV2SR model prices are
observed not to converge to the analytical BS price; i.e. their differences to
the analytical price are larger in magnitude than the MC errors.
This shows us the impact of stochastic volatility on the barrier option price.

While the stochasticity of the local volatility has a clear impact on the price
of path dependent instruments, the stochasticity of interest rates have little
effect under standard market conditions, e.g. when the interest rate
volatilities are low. We expect this effect to be more prominent in
stressed environments with higher interest rate volatilities, which is what we
study next.

Consider the stochastic rates extension of the Black Scholes model (BS2SR)
\begin{equation}
\begin{split}
dS_t =& \left[r^d_t - r^f_t \right] S_t dt + \sigma^S S_t
dW^{S\text{(DRN)}}_t,\\
dx^d_t =& -a^d_t x^d_t dt + \sigma^d_t dW^{d\text{(DRN)}}_t,\ r^d_t = x^d_t +
\phi^d_t,\\
dx^f_t =& \left[ -a^f_t x^f_t - \rho_{Sf} \sigma^f_t \sigma^S \right] dt +
\sigma^f_t dW^{f\text{(DRN)}}_t,\ r^f_t = x^f_t + \phi^f_t.
\label{eqn:BS2SR_SDEs_DRN}
\end{split}
\end{equation}
This model can be seen as a special case of LV2SR with flat local volatility,
which allows us to incorporate results from Section \ref{sec:LV2SR}. Using
(\ref{eqn:ersds}) and It\^{o}'s lemma one can write the SDEs for the
domestic and foreign zero coupon bonds in domestic risk neutral measure as
\begin{equation}
\begin{split}
\frac{dP^d(t, T)}{P^d(t, T)} =& r^d_t dt - \sigma^d_t b^d(t, T) dW^{d\text{(DRN)}}_t,\\
\frac{dP^f(t, T)}{P^f(t, T)} =& \left[r^f_t + \rho_{Sf} \sigma^S \sigma^f_t
b^f(t, T)\right] dt - \sigma^f_t b^f(t, T) dW^{f\text{(DRN)}}_t.
\end{split}
\end{equation}
Using (\ref{eqn:RN_dT_dDRN}) and Lemma \ref{lemma:measure_trf}, these can be
written in the domestic $T$-forward measure as
\begin{equation}
\begin{split}
\frac{dP^d(t, T)}{P^d(t, T)} = & \left[r^d_t + \left(\sigma^d_t b^d(t, T)
 \right)^2 \right] dt - \sigma^d_t b^d(t, T)
 dW^{d\text{(T)}}_t,\\
 \frac{dP^f(t, T)}{P^f(t, T)} =& \left[r^f_t +  \sigma^f_t b^f(t,
 T)\left(\rho_{Sf} \sigma^S + \rho_{df} \sigma^d_t b^d(t, T)\right)\right] dt\\
 &-\sigma^f_t b^f(t, T) dW^{f\text{(T)}}_t.
 \label{eqn:dPdPf_Tfwd}
\end{split}
\end{equation}
Similarly, the exchange rate process written in the domestic $T$-forward
measure is given by (c.f. (\ref{eqn:LV2SR_ST}))
\begin{equation}
dS_t = \left[ r^d_t - r^f_t - \rho_{Sd} b^d(t, T) \sigma^d_t \sigma^S
\right] S_t dt + \sigma^S S_t dW^{S\text{(T)}}_t.
\label{eqn:dS_Tfwd}
\end{equation}
The forward value of the exchange rate is
\begin{equation}
F(t, T) \equiv \mathbf{E}^{\mathbb{Q}^{\text{(T)}}}[S_T \mid \mathcal{F}_t] =
S_t \frac{P^f(t, T)}{P^d(t, T)},
\end{equation}
which is a martingale under the $T$-forward measure $\mathbb{Q}^{\text{(T)}}$.
Its SDE can be computed from (\ref{eqn:dPdPf_Tfwd}), (\ref{eqn:dS_Tfwd}),
and application of It\^{o}'s lemma,
\begin{equation}
\frac{dF(t, T)}{F(t, T)} = \sigma^S dW^{S\text{(T)}}_t+ \sigma^d_t b^d(t,
T) dW^{d\text{(T)}}_t - \sigma^f_t b^f(t, T) dW^{f\text{(T)}}_t.
\end{equation}
Since the diffusion process above is a linear combination of correlated
Brownian motions, we can extract the total implied variance easily,
\begin{equation}
\begin{split}
\Sigma^2 T =& (\sigma^S)^2 T + 2\sigma^S \int_0^T \left[\rho_{Sd} \sigma^d_t
b^d(t, T) - \rho_{Sf} \sigma^f_t b^f(t, T)\right]dt\\
&+\int_0^T \left[\big(\sigma^d_t b^d(t, T)\big)^2 - 2\rho_{df} \sigma^d_t
b^d(t, T) \sigma^f_t b^f(t, T) + \big(\sigma^f_t b^f(t, T)\big)^2 \right]dt.
\label{eqn:tivF}
\end{split}
\end{equation}
With this quantity one can write down the Black formula for the time zero value
of a vanilla call option as
\begin{equation}
C(K, T) = P^d(0, T)\left[F(0, T) N(\tilde{d}_1) - K
N(\tilde{d}_2) \right],
\end{equation}
with $\tilde{d}_1 = \frac{\log\frac{F(0, T)}{K} + \frac{1}{2} \Sigma^2
T}{\Sigma\sqrt{T}}$ and $\tilde{d}_2 = \tilde{d}_1 - \Sigma\sqrt{T}$.

The integrals in (\ref{eqn:tivF}) can be evaluated numerically. Therefore, for
a given market implied volatility $\Sigma$ at a given maturity $T$ and strike
$K$, one can solve this quadratic equation to find the BS2SR volatility
$\sigma^S$ that will reproduce the market quotes.
Conversely, the right hand side of (\ref{eqn:tivF}) dictates the lower bound
of total implied variance for which there is a solution for the BS2SR
model. The lower bound can be evaluated, by taking the derivative with respect
to $\sigma^S$ and setting it to zero,
\begin{equation}
\begin{split}
\Sigma_{\text{min}}^2 T =& \int_0^T \left[\big(\sigma^d_t b^d(t, T)\big)^2 -
2\rho_{df} \sigma^d_t b^d(t, T) \sigma^f_t b^f(t, T) + \big(\sigma^f_t b^f(t,
T)\big)^2 \right]dt\\
&-\frac{1}{T} \left(\int_0^T \left[\rho_{Sd} \sigma^d_t
b^d(t, T) - \rho_{Sf} \sigma^f_t b^f(t, T)\right]dt\right)^2.
\end{split}
\end{equation}
If the market
total implied variance at a given maturity $T$ and strike $K$ is lower than
this, the BS2SR model will not have a solution. As a consequence, the local
volatility extension of this model (LV2SR) will not be calibratable, that is the
evaluation of (\ref{eqn:dupire_C_stochastic_ir}) during the application of the
calibration algorithm will lead to imaginary values for local volatility given
the already fixed parameters of the interest rate models and the correlations.
This signifies that no real and positive local volatility exists that would
reproduce the market quotes for vanilla options within such a model.

As a study, we take the market data as of 2020-04-30, but vary the interest rate 
model parameters and correlations, and compare the minimum total implied
variance admitted by the BS2SR model for a number of values of interest rate 
volatilities and correlations. In particular for a mean reversion $a^d_t = a^f_t = 0.01$, and volatilities
$\sigma^d_t = \sigma^f_t = 0.02, 0.05$ we look at the minimum total implied
variances admitted by the BS2SR model for various values of $\rho_{df}$.
\begin{figure}[ht!]
    \centering \includegraphics[width=\textwidth]{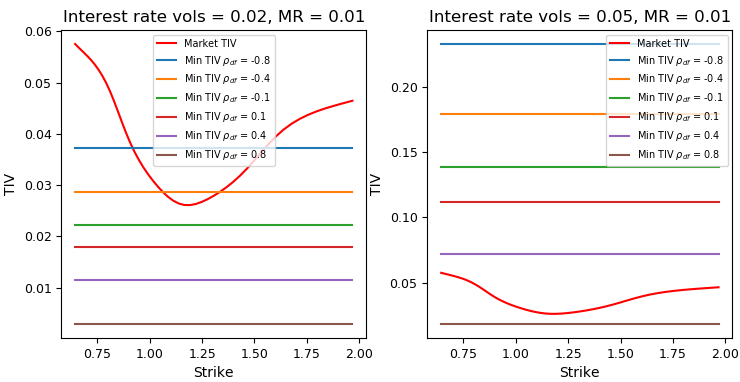}
    \caption{Minimum total implied variances (TIV) allowed by the BS2SR model
    for various values of interest rate volatilities and correlations compared
    to market TIV.
    $\rho_{Sd} = 0.166$, $\rho_{Sf} = 0.551$}
    \label{plot:BS2SR_tiv}
\end{figure}
Figure \ref{plot:BS2SR_tiv} shows that for interest rate volatilities set to
0.02, the market total implied variance is lower than the minimum total implied
variance admitted by the BS2SR model at various strikes when the correlation
$\rho_{df}$ is below -0.4. In case the interest rate volatilities are set to
0.05, the market total implied variance is attainable only when the correlation
$\rho_{df}$ is quite high, around 0.8.

\appendix
\section{Supplementary computations}\label{sec:intrsds}

The following finding is utilized in measure transformations throughout the
paper.
\begin{lemma}\label{lemma:measure_trf}
Let
\begin{equation}
dX_t = a(X_t, t) dt + b(X_t, t) dW^{X\text{(A)}}_t, \label{eqn:dX_A}
\end{equation}
with
\begin{equation*}
P\left(\int_0^t \left(\left| a(X_s, s) \right| + \left| b^2(X_s, s)
\right|\right) ds < \infty\right) = 1,\ \forall t \geq 0,
\end{equation*}
be an SDE describing the evolution of process $(X_t)$
under measure $\mathbb{Q}^{\text{A}}$ whose Brownian motion
$(W^{X\text{(A)}}_t)$ is correlated to another Brownian motion
$(W^{Y\text{(A)}}_t)$ with a coefficient of correlation $\rho_{XY}$, that is
$d\left<W^{X\text{(A)}}, W^{Y\text{(A)}} \right>_t = \rho_{XY} dt$.
Let $(W^{Y\text{(B)}}_t)$
be a Brownian motion under an equivalent measure $\mathbb{Q}^{\text{B}}$
characterized by the transformation
\begin{equation}
\frac{d\mathbb{Q}^{\text{B}}}{d\mathbb{Q}^{\text{A}}} =
\exp\left[ -\frac{1}{2}\int_0^t
c^2(\cdot, s) ds -
\int_0^t c(\cdot, s) dW^{Y\text{(A)}}_s\right],
\end{equation}
such that
\begin{equation}
dW^{Y\text{(B)}}_t = dW^{Y\text{(A)}}_t + c(\cdot, t) dt,
\label{eqn:dW_YA_YB}
\end{equation}
with $P\left(\int_0^t \left| c(\cdot, s) \right| ds < \infty\right) = 1,\
\forall t \geq 0$, and $c(\cdot, t)$ is a function of underlying assets of
the SDE system, and time.
Then the evolution of $(X_t)$ under measure $\mathbb{Q}^{\text{B}}$ is described
by
\begin{equation}
dX_t = \left[a(X_t, t) - \rho_{XY} b(X_t, t) c(\cdot, t) \right] dt + b(X_t, t)
dW^{X\text{(B)}}_t,
\end{equation}
where $(W^{X\text{(B)}}_t)$ is a Brownian motion under $\mathbb{Q}^{\text{B}}$.
\end{lemma}

\begin{proof}
We can decompose the Brownian motion $(W^{X\text{(A)}}_t)$ into
$(W^{Y\text{(A)}}_t)$ and an independent Brownian motion $(Z_t)$, that
is $
d\left<W^{Y\text{(A)}}, Z \right>_t = 0,
$
as
\begin{equation}
dW^{X\text{(A)}}_t = \rho_{XY} dW^{Y\text{(A)}}_t + \sqrt{1 - \rho_{XY}^2}
dZ_t. \label{eqn:dW_XA_YA}
\end{equation}
We note that, as a result of the multi-dimensional Girsanov theorem, $(Z_t)$ is
a Brownian motion under $\mathbb{Q}^{\text{B}}$.

Now we can use (\ref{eqn:dW_YA_YB}) and (\ref{eqn:dW_XA_YA}) to write the
process (\ref{eqn:dX_A}) as
\begin{equation}
\begin{split}
dX_t =& a(X_t, t) dt + b(X_t, t) \left[\rho_{XY} dW^{Y\text{(A)}}_t + \sqrt{1 -
\rho_{XY}^2} dZ_t\right]\\
=& a(X_t, t) dt + b(X_t, t) \left[\rho_{XY} \left( dW^{Y\text{(B)}}_t - c(\cdot,
t) dt \right) + \sqrt{1 - \rho_{XY}^2} dZ_t\right]\\
=& \left[ a(X_t, t) - \rho_{XY} b(X_t, t) c(\cdot, t) \right] dt +
b(X_t, t) dW^{X\text{(B)}}_t, \label{eqn:dX_B}
\end{split}
\end{equation}
with
\begin{equation}
\begin{split}
dW^{X\text{(B)}}_t =& \rho_{XY} dW^{Y\text{(B)}}_t + \sqrt{1 - \rho_{XY}^2}
dZ_t\\
=& dW^{X\text{(A)}}_t + \rho_{XY} c(\cdot, t) dt.
\end{split}
\end{equation}
\end{proof}

We use the following result during the change to $T$-forward measure in Section
\ref{sec:Tforward}. The limiting case with constant coefficients was
investigated in \cite{BrigoMercurio2013}. Here we study the general case with
time dependent coefficients.

\begin{lemma}\label{lemma:g1pp_identity}
For the G1++ model (\ref{eqn:drd_G1PP_DRN}) describing the evolution of the
short rate $(r_t)$,
\begin{equation}
\begin{split}
r_t &= x_t + \phi_t, \\
dx_t &= -a_t x_t dt + \sigma_t dW_t,
\end{split}
\end{equation}
where $\phi_t$ is the deterministic shift function that is calibrated to market
yield curve; $a_t$ is the mean reversion coefficient, and $\sigma_t$ is the
volatility coefficient, the following identity holds,
\begin{equation}
\exp{\left[-\int_t^T r_s ds\right]} = P(t, T) \exp\left[-\int_t^T \sigma_v b(v,
T) dW_v - \frac{1}{2} \int_t^T \sigma_v^2 b^2(v, T) dv \right].
\label{eqn:ersds}
\end{equation}
Here, $P(t, T) \equiv \mathbf{E}\left[e^{-\int_t^T r_s ds} \big\vert
\mathcal{F}_t \right]$ is the time $t$ value of a zero coupon bond maturing at
time $T$, and $b(t, T) \equiv \int_t^T e^{- \int_t^v a_z dz} dv$.

\end{lemma}

\begin{proof}
The integral on the left hand side of (\ref{eqn:ersds}) can be split into
$-\int_t^T r_s ds = -\int_t^T x_s ds -\int_t^T \phi_s ds$. We start with
integrating the first of these integrals by parts,
\begin{equation}
\begin{split}
\int_t^T x_s ds =& s x_s \big\vert_t^T - \int_t^T s dx_s\\
=& (T-t) x_t + \int_t^T (T-v) (-a_v x_v dv + \sigma_v dW_v).
\label{eqn:int_xudu_byparts}
\end{split}
\end{equation}
We compute $x_s$ by evaluating the following integral
\begin{equation*}
\int_t^s d_u\left(x_u e^{\int_t^u a_z dz}\right) =
\int_t^s a_u e^{\int_t^u a_z dz} x_u du +
\int_t^s e^{\int_t^u a_z dz} (-a_u x_u du + \sigma_u dW_u),
\end{equation*}
which leads to
\begin{equation}
x_s = x_t e^{-\int_t^s a_z dz} + \int_t^s e^{-\int_u^s a_z dz} \sigma_u dW_u.
\end{equation}
We plug this into (\ref{eqn:int_xudu_byparts}) to get
\begin{equation*}
\begin{split}
\int_t^T x_s ds =& (T-t) x_t - x_t \int_t^T (T-v) a_v e^{-\int_t^v a_z dz} dv\\
&+ \int_t^T (T-v) \left[
-a_v \left( \int_t^v e^{-\int_u^v a_z dz}
\sigma_u dW_u \right) dv + \sigma_v dW_v \right].
\end{split}
\end{equation*}
The first two integrals appearing in the right hand side above can be evaluated
by integration by parts. The first one yields
\begin{equation*}
\begin{split}
-\int_t^T (T-v) a_v e^{-\int_t^v a_z dz} dv =& (T-v) e^{-\int_t^v a_z
dz}\big\vert_{v=t}^T + \int_t^T e^{-\int_t^v a_z dz} dv \\
=& -(T-t) + b(t, T).
\end{split}
\end{equation*}
The second one evaluates
\begin{equation*}
\begin{split}
-&\int_t^T (T-v) a_v \int_t^v e^{-\int_u^v a_z dz} \sigma_u dW_u dv \\
=&-\int_t^T \left(\int_t^v e^{\int_t^u a_z dz} \sigma_u dW_u \right)
d_v\left(\int_t^v a_y (T-y) e^{-\int_t^y a_z dz} dy\right)\\
=& -\left[\int_t^v e^{\int_t^u a_z dz} \sigma_u dW_u \int_t^v a_y (T-y)
e^{-\int_t^y a_z dz} dy \right]_{v=t}^T\\
&+\int_t^T e^{\int_t^v a_z dz} \sigma_v dW_v \int_t^v a_y (T-y) e^{-\int_t^y a_z
dz} dy\\
=& - \int_t^T e^{\int_t^v a_z dz} \sigma_v dW_v \int_v^T a_y (T-y)
e^{-\int_t^y a_z dz} dy\\
=& \int_t^T \left[-(T-v) + b(v, T) \right] \sigma_v dW_v.
\end{split}
\end{equation*}
This leads to
\begin{equation}
\int_t^T x_s ds = x_t b(t, T) + \int_t^T \sigma_v b(v, T) dW_v.
\label{eqn:int_xu_du}
\end{equation}
Hence
\begin{equation*}
\begin{split}
\mathbf{E}\left[\int_t^T x_u du\right] =& x_t b(t, T)\\
\mathbf{Var}\left[\int_t^T x_u du\right] =& \int_t^T \sigma_v^2 b^2(t, T) dv.
\end{split}
\end{equation*}
Since $\mathbf{E}\left[e^{-X}\right] = e^{-\mu + \frac{1}{2}\sigma^2}$ where
$X \sim N(\mu, \sigma^2)$, we have
\begin{equation*}
\begin{split}
P(t, T) =& \mathbf{E}\left[e^{-\int_t^T (\phi_s + x_s) ds} \big\vert
\mathcal{F}_t \right] = e^{-\int_t^T \phi_s ds - x_t b(t, T) + \frac{1}{2}
\int_t^T \sigma_v^2 b^2(v, T) dv}\\
=& e^{\int_t^T x_s ds} e^{-\int_t^T r_s ds} e^{- x_t b(t, T) + \frac{1}{2}
\int_t^T \sigma_v^2 b^2(v, T) dv}.
\end{split}
\end{equation*}
Plugging in (\ref{eqn:int_xu_du}) into this expression gives the desired result.
\end{proof}

\begin{lemma}\label{lemma:dupire_derivation}
Let
\begin{equation}
dS_t = (r^d_t - r^f_t) S_t dt + \sigma_S(t, S_t, U_t) S_t dW^{S\text{(DRN)}}_t 
\label{eqn:general_sde}
\end{equation}
be an SDE describing the evolution of asset process $(S_t)$ under the domestic
risk neutral measure $\mathbb{Q}^{\text{DRN}}$ associated with the money market
numer\'aire $B_t^d = \exp[\int_0^t r_s^d ds]$; $(r^d_t)$, $(r^f_t)$ and $(U_t)$
be stochastic processes.
Let $C(K, T)$ be the price of a vanilla call option written on $(S_t)$, with
strike $K$ and maturity $T$, so that $C(K, T) =
\mathbf{E}^{\mathbb{Q}^{\text{DRN}}}[\frac{1}{B_T^d}(S_T - K)^+]$.
Then the following identity holds
\begin{equation*}
\mathbf{E}^{\mathbb{Q}^{\text{T}}}\left[ \sigma^2_S(T,S_T,U_T) | S_T=K  \right]
= \frac{\frac{\partial C}{\partial T} - P^d(0, T)
\mathbf{E}^{\mathbb{Q}^{\text{T}}}\left[ (r^d_T K - r^f_T S_T )
\mathds{1}_{S_T>K} \right]}{ \frac{1}{2} K^2 \frac{\partial^2 C}{\partial K^2}},
\label{eqn:dupire_intermediate}
\end{equation*}
where the expectations are taken under the $T$-forward measure
$\mathbb{Q}^{\text{T}}$ associated with the zero coupon bond price numer\'aire
$P^d(0, T) = \mathbf{E}^{\mathbb{Q}^{\text{DRN}}}[\frac{1}{B_T^d}]$.
\end{lemma}

\begin{proof}
Here we apply the methodology from Section 10.2.1 of \cite{OG2019} to our setup.
The Tanaka-Meyer formula for the positive part function for a continuous
semimartingale states in differential form
\begin{equation*}
d(S_t-K)^+ = \mathds{1}_{S_t>K} dS_t + \frac{1}{2} \delta(S_t-K) d\langle S\rangle_t ,
\end{equation*}
with $\langle\cdot\rangle_s$ being the quadratic variation, and $\delta$ being
the Dirac delta distribution, possibly characterized as a limit.
We will later take expectations, so we are only interested in the drift term.
Applying (\ref{eqn:general_sde}) into this, we obtain that the drift term of $d(S_t-K)^+$
as
\begin{equation*}
\mathds{1}_{S_t>K} (r^d_t - r^f_t) S_t + \frac{1}{2} \delta(S_t-K)
\sigma^2_S(t,S_t,U_t) S^2_t .
\end{equation*}
We use the product rule to compute the drift term
of the discounted payoff $d\left[\frac{1}{B^d_t}(S_t-K)^+\right]$ as
\begin{equation*}
\frac{1}{B^d_t} \left( 
 -r^d_t (S_t-K)^+ + \mathds{1}_{S_t>K} (r^d_t - r^f_t) S_t + \frac{1}{2} \delta(S_t-K) \sigma^2_S(t,S_t,U_t) S^2_t 
\right),
\end{equation*}
which we simplify, using $(S_t-K)^+ = \mathds{1}_{S_t>K} (S_t-K)$, to
\begin{equation*}
\frac{1}{B^d_t} \left( 
 \mathds{1}_{S_t>K} (r^d_t K - r^f_t S_t ) + \frac{1}{2} \delta(S_t-K) \sigma^2_S(t,S_t,U_t) S^2_t 
\right).
\end{equation*}
We take the expectation under the domestic risk neutral measure
$\mathbf{E}^{\mathbb{Q}^{\text{DRN}}}$, noting the fact that the expression
inside the differential is the value of the vanilla call option with strike $K$
and maturity $T$, and assuming the diffusion term is a true martingale, to
arrive at
\begin{equation*}
\begin{split}
\frac{\partial C}{\partial T} = & \mathbf{E}^{\mathbb{Q}^{\text{DRN}}}\left[
\frac{1}{B^d_T} \left( \mathds{1}_{S_T>K} (r^d_T K - r^f_T S_T ) + \frac{1}{2}
\delta(S_T-K) \sigma^2_S(T,S_T,U_T) S^2_T \right) \right]\\
 = & P^d(0, T)
\mathbf{E}^{\mathbb{Q}^{\text{T}}}\left[
\mathds{1}_{S_T>K} (r^d_T K - r^f_T S_T ) + \frac{1}{2} \delta(S_T-K)
\sigma^2_S(T,S_T,U_T) S^2_T \right]
\end{split}
\end{equation*}
We first consider 
\begin{equation*}
\begin{split}
\mathbf{E}^{\mathbb{Q}^{\text{T}}} \left[ \delta(S_T-K) \sigma^2_S(T,S_T,U_T)
S^2_T \right]  = &
\mathbf{E}^{\mathbb{Q}^{\text{T}}} \left[ \mathbf{E}^{\mathbb{Q}^{\text{T}}}
\left[ \sigma^2_S(T,S_T,U_T) | S_T\right]  \delta(S_T-K) S^2_T  \right] \\
 = & q^T(K, T) K^2 \mathbf{E}^{\mathbb{Q}^{\text{T}}}\left[
 \sigma^2_S(T,S_T,U_T) | S_T=K \right]
\end{split}
\end{equation*}
where we used the properties of conditional expectations, and denoted the
marginal distribution of $(S_T)$ under the $T$-forward measure with $q^T(S_T,
T)$. Now we use \cite{BL1978} relationship $\frac{\partial^2
C}{\partial K^2} = P^d(0,T) q^T(K, T)$  to rewrite
\begin{equation*}
 \begin{split}
\frac{\partial C}{\partial T} = &P^d(0,T)
\mathbf{E}^{\mathbb{Q}^{\text{T}}}\left[ \mathds{1}_{S_T>K} (r^d_T K - r^f_T S_T
) \right] \\
&+ \frac{1}{2} K^2 \frac{\partial^2 C}{\partial K^2}
\mathbf{E}^{\mathbb{Q}^{\text{T}}}\left[ \sigma^2_S(T,S_T,U_T) | S_T=K  \right]
\end{split}   
\end{equation*}
\end{proof}

For the LV2SR model, $\sigma_S(t,S_t,U_t) = \sigma_{LV} (S_t,t)$, thus the
second expectation simplifies to $\sigma^2_{LV} (K,T)$ and we obtain
\begin{equation}
\frac{\partial C}{\partial T} = P^d(0,T)
\mathbf{E}^{\mathbb{Q}^{\text{T}}}\left[ \mathds{1}_{S_T>K} (r^d_T K - r^f_T S_T ) \right]  + 
\frac{1}{2} K^2 \frac{\partial^2 C}{\partial K^2}  \sigma^2_{LV} (K,T).
\label{eqn:Dupire_pre_LV2SR}
\end{equation}
For the SLV2SR model, $\sigma_S^2(t,S_t,U_t) = L^2(S_t,t) U_t$, thus the second
expectation factors to $L^2(K,T) \mathbf{E}^{\mathbb{Q}^{\text{T}}}\left[ U_T |
S_T=K  \right]$ and we obtain
\begin{equation}
 \begin{split}
\frac{\partial C}{\partial T} = & P^d(0,T)
\mathbf{E}^{\mathbb{Q}^{\text{T}}}\left[ \mathds{1}_{S_T>K} (r^d_T K - r^f_T S_T
) \right]\\
& +\frac{1}{2} K^2 \frac{\partial^2 C}{\partial K^2} L^2(K,T) 
\mathbf{E}^{\mathbb{Q}^{\text{T}}}\left[ U_T | S_T=K  \right].
\label{eqn:Dupire_pre_SLV2SR}
\end{split}   
\end{equation}
We can rewrite (\ref{eqn:Dupire_pre_LV2SR}) as
\begin{equation}
\sigma^2_{LV} (K,T)  = \frac{\frac{\partial C}{\partial T}  - P^d(0,T)
\mathbf{E}^{\mathbb{Q}^{\text{T}}}\left[(r^d_T K - r^f_T S_T )
\mathds{1}_{S_T>K} \right]  } { \frac{1}{2} K^2 \frac{\partial^2 C}{\partial
K^2}}.\label{eqn:Dupire_LV2SR}
\end{equation}
In the total implied variance parametrization of the Black-Scholes model, using
identities (\ref{eqn:CBS_first_derivatives}) together with
\begin{equation}
\begin{split}
\frac{\partial^2 C_{\text{BS}}}{\partial w^2} =& \frac{1}{2}\frac{\partial
C_{\text{BS}}}{\partial w}\left[-\frac{1}{4} -\frac{1}{w} +
\frac{y^2}{w^2}\right],\\
\frac{\partial^2 C_{\text{BS}}}{\partial w \partial y} =& \frac{\partial
C_{\text{BS}}}{\partial w} \left[-\frac{y}{w} + \frac{1}{2}\right],\\
\frac{\partial^2 C_{\text{BS}}}{\partial y^2} =& \frac{\partial
C_{\text{BS}}}{\partial y} + 2 \frac{\partial C_{\text{BS}}}{\partial w},
\end{split}
\end{equation}
one derives
\begin{equation}
\begin{split}
K^2 \frac{\partial^2 C_{\text{BS}}}{\partial K^2} =&
\frac{\partial^2 C_{\text{BS}}}{\partial y^2}
+ \left(2\frac{\partial^2 C_{\text{BS}}}{\partial w \partial y} 
+ \frac{\partial^2 C_{\text{BS}}}{\partial w^2} \frac{\partial w}{\partial y}
- \frac{\partial C_{\text{BS}}}{\partial w}
\right) \frac{\partial w}{\partial y}\\
&+ \frac{\partial C_{\text{BS}}}{\partial w} \frac{\partial^2 w}{\partial y^2}
-\frac{\partial C_{\text{BS}}}{\partial y}\\
=&2\frac{\partial C_{\text{BS}}}{\partial w} \left[
1 - \frac{y}{w} \frac{\partial w}{\partial y}
+ \frac{1}{2} \frac{\partial^2 w}{\partial y^2}
+ \frac{1}{4} \left(\frac{\partial w}{\partial y}\right)^2
\left(-\frac{1}{4}- \frac{1}{w} + \frac{y^2}{w^2}\right)
\right].
\end{split}
\end{equation}
Therefore in this parametrization (\ref{eqn:Dupire_LV2SR}) can be written as
\begin{equation}
\sigma^2_{LV} (K,T)  = \frac{\frac{\partial C_{\text{BS}}}{\partial T}  - P^d(0,T)
\mathbf{E}^{\mathbb{Q}^{\text{T}}}\left[(r^d_T K - r^f_T S_T )
\mathds{1}_{S_T>K} \right]  } { \frac{\partial C_{\text{BS}}}{\partial w} \left[
1 - \frac{y}{w} \frac{\partial w}{\partial y}
+ \frac{1}{2} \frac{\partial^2 w}{\partial y^2}
+ \frac{1}{4} \left(\frac{\partial w}{\partial y}\right)^2
\left(-\frac{1}{4}- \frac{1}{w} + \frac{y^2}{w^2}\right)
\right]},\label{eqn:Dupire_LV2SR_tiv}
\end{equation}
where $\frac{\partial C_{\text{BS}}}{\partial T}$ and $\frac{\partial
C_{\text{BS}}}{\partial w}$ are as given in (\ref{eqn:CBS_first_derivatives}).

Notice that (\ref{eqn:Dupire_LV2SR}) can be computed for any SDE system that
includes the given SDE subsystems for $(S)$, $(r^d)$ and $(r^f)$. In particular,
one can compute such a ``local volatility'' even for stochastic local volatility
models. (\ref{eqn:Dupire_pre_SLV2SR}) can be reformulated as \cite{Ogetbil2020}
\begin{equation}
L^2(K,T) \mathbf{E}^{\mathbb{Q}^{\text{T}}}\left[ U_T | S_T=K  \right]  =
\frac{\frac{\partial C}{\partial T} -  P^d(0,T) \mathbf{E}^{\mathbb{Q}^{\text{T}}}\left[(r^d_T K
- r^f_T S_T )  \mathds{1}_{S_T>K}  \right]  } { \frac{1}{2} K^2 \frac{\partial^2
C}{\partial K^2}},\label{eqn:Dupire_SLV2SR}
\end{equation}
and by comparing this to (\ref{eqn:Dupire_LV2SR}) we arrive at the relationship
between the SLV2SR model leverage function and the LV2SR submodel local
volatility function,
\begin{equation}
L^2(K,T)  =  \frac {\sigma^2_{LV} (K,T)  }
{\mathbf{E}^{\mathbb{Q}^{\text{T}}}\left[ U_T | S_T=K  \right]
}.\label{eqn:leverage_localvol_relation}
\end{equation}

\paragraph{Acknowledgments}
The authors are grateful to the reviewers and the editor for their valuable
comments that have helped us improve the paper substantially.
The authors are indebted to Dooheon Lee and Kisun Yoon for numerous enlightening
discussions and guidance about the theoretical foundations of this paper.
The authors would also like to thank Agus Sudjianto for
supporting this research, and Vijayan Nair for suggestions, feedback, and
discussion regarding this work. Any opinions, findings and conclusions or
recommendations expressed in this material are those of the authors and do not
necessarily reflect the views of Wells Fargo Bank, N.A., its parent company,
affiliates and subsidiaries.

\newpage{}
\bibliographystyle{unsrt}
\bibliography{mixed}

\end{document}